\documentclass[journal]{IEEEtran}

\usepackage[dvipsnames]{xcolor}
\usepackage[normalem]{ulem}
\usepackage[utf8x]{inputenc}
\usepackage{balance}
\usepackage{graphicx}
\usepackage{epstopdf}
\usepackage{amsmath}
\usepackage{subfigure}
\usepackage{float}
\usepackage{cite}
\usepackage{amssymb}
\usepackage{mathrsfs}
\usepackage{amsthm}
\usepackage{multirow}

\ifCLASSINFOpdf

\else

\fi

\newtheorem{theorem}{Theorem}

\begin{document}

\title{Achieving Covertness and Secrecy: A New Paradigm for Secure Wireless Communication}

\author{Huihui Wu, 
	Yuanyu Zhang, \IEEEmembership{Member, IEEE},
	Yulong Shen,  \IEEEmembership{Member, IEEE} and
	Xiaohong Jiang, \IEEEmembership{Senior Member, IEEE}
	\thanks{H. Wu is with the School of Computer Science and Technology, Xidian University, Xi'an, Shaanxi, China, and also with the School of Systems Information Science, Future University Hakodate, Hakodate, Hokkaido, Japan (Emails: hhwu2015@163.com).}
	\thanks{Y. Zhang is with the School of Computer Science and Technology, Xidian University, Xi'an, Shaanxi, China, and also with the Graduate School of Science and Technology, Nara Institute of Science and Technology, 8916-5 Takayama, Ikoma, Nara, 630-0192, Japan (Email: yy90zhang@ieee.org).}
	\thanks{Y. Shen is with the School of Computer Science and Technology, Xidian University, Xi'an, Shaanxi, China (Email: ylshen@mail.xidian.edu.cn).}
	\thanks{X. Jiang is with the School of Systems Information Science, Future University Hakodate, Hakodate, Hokkaido, Japan (Emails: jiang@fun.ac.jp).}
}

\maketitle

\begin{abstract}

This paper explores a new secure wireless communication paradigm where the physical layer security technology is applied to counteract both the detection and eavesdropping attacks, such that the critical covertness and secrecy properties of the communication are jointly guaranteed. We first provide theoretical modeling for covertness outage probability (COP), secrecy outage probability (SOP) and transmission probability (TP) to depict the covertness, secrecy and transmission performances of the paradigm. To understand the fundamental security performance under the new paradigm, we then define a new metric - covert secrecy rate (CSR), which characterizes the maximum transmission rate subject to the constraints of COP, SOP and TP. We further conduct detailed theoretical analysis to identify the CSR under various scenarios determined by the detector-eavesdropper relationships and the secure transmission schemes adopted by transmitters. Finally, numerical results are provided to illustrate the achievable performances under the new secure communication paradigm.

\end{abstract}


\IEEEpeerreviewmaketitle

\section{Introduction}

\IEEEPARstart{T}{he} fundamental research of wireless communication security is of great importance for the development of secure network communication, information security and communication privacy \cite{zou2016survey,djenouri2020survey}. It is notable that in modern secure wireless communication applications, covertness and secrecy serve as two typical properties. Covertness concerns with the protection of wireless communication from detection attacks that attempt to detect the existence of the communication \cite{bash2015hiding,yan2019low}, while secrecy deals with the protection of wireless communication from eavesdropping attacks \cite{wang2018survey,hamamreh2018classifications} which manage to intercept the information conveyed by the communication. With the wide application of secure wireless communication, how to ensure the covertness and secrecy of such communication has become an increasingly urgent demand.

Thanks to the rapid progress of information and communication technologies, physical layer security (PLS) technique is now regarded as a highly promising approach to counteract the detection and eavesdropping attacks and thus to ensure the covertness and secrecy properties of wireless communications. The basic principle behind the PLS technology is to exploit the inherent physical layer randomness of wireless channels (e.g., noise and fading) to implement the secure and covert communications \cite{liu2016physical}. For example, transmitters can intentionally inject artificial noise (AN) into their channels to hide their signals from detectors or to add uncertainty to the information intercepted by eavesdroppers. The PLS technology realizes secure wireless communications from the information-theoretic perspective and thus provides stronger form of covertness and secrecy guarantees than traditional security technologies like the cryptography and spread spectrum \cite{liang2009information,simon2002spread,forouzan2007cryptography}. Actually, the PLS technology serves as an effective supplement for the traditional security technologies to significantly improve the covertness and secrecy of wireless communications \cite{yan2019low,mukherjee2014principles}.

By now, extensive research efforts have been devoted to study of covertness or secrecy guarantee for wireless communication based on the PLS technology. In \cite{hu2018covert,forouzesh2019robust,shahzad2018achieving,zheng2020covert,sobers2017covert,li2020optimal}, the AN technique or cooperative jamming technique was adopted for covert wireless communication in the typical three-node scenario with a transmitter, a receiver and a malicious detector. In these works, the AN may be initiated by the transmitter \cite{hu2018covert,forouzesh2019robust}, by the (full-duplex) receiver \cite{shahzad2018achieving,zheng2020covert}, or by some external helper nodes \cite{sobers2017covert,li2020optimal} to avoid the communication signal from being detected by the detector. The works in \cite{lee2018covert,forouzesh2019robust,goeckel2016covert,he2017covert} show that the covert wireless communication can be implemented by exploiting the detector's uncertainty about its channel state information, like the instantaneous channel coefficient \cite{lee2018covert}, statistical channel coefficient \cite{forouzesh2019robust} or background noise \cite{goeckel2016covert,he2017covert}. Such uncertainty makes it difficult for the detector to determine the received signal power or the background noise power, and thus unable to distinguish between the scenarios with or without wireless communication by examining the power difference in these scenarios. Some recent works also explored the possibility of ensuring covertness based on other PLS technologies, such as multi-antenna technique \cite{shahzad2019covert,zheng2019multi}, coding scheme \cite{kibloff2019embedding,tahmasbi2019covert}, relay selection \cite{gao2021covert,su2020covert} and resource (i.e., channel use) allocation \cite{sun2020resource}.

The PLS technology has also been widely adopted for achieving secrecy in various wireless communication scenarios, such as ad-hoc networks \cite{zhang2021secure,sarkohaki2020efficient}, device-to-device (D2D) communications \cite{zhang2018mode,khoshafa2020secure}, cellular networks \cite{peng2021multiuser,abbas2020analysis} and the Internet of Things (IoT) \cite{zhang2015secure,khan2020efficient}. These works mainly exploited the application of AN technique to create a relatively better channel to the receiver than that to the eavesdropper with the aim of achieving a positive secrecy rate. In \cite{zhao2020secure,huanga2021transmit}, the beamforming technique was explored for secure wireless communication in multi-antenna scenarios, where the transmit power of signals was concentrated toward the direction of intended receiver such that a much better signal quality at the receiver can be created than that at the eavesdropper. The work in \cite{wang2020energy} further combined the beamforming and AN techniques to achieve a significant signal advantage at the receiver, while the works in \cite{he2019link,feng2018two} considered the multi-user scenarios and applied relay selection technique to create a transmitter-receiver channel advantage over the transmitter-eavesdropper channel. Some other works in \cite{liu2020secrecy,wu2018secrecy,wu2019energy} also studied the secure wireless communication based on the technique of resource allocation (e.g., power allocation, time slot allocation, energy allocation).

The above works help us understand the great potentials of the PLS technology in ensuring the covertness or secrecy of wireless communication. It is notable that these works mainly focus on the traditional paradigms of secure wireless communication where only one type of attack may exist, be it detection or eavesdropping, and concern with either the covertness guarantee or secrecy guarantee for wireless communications. In practice, however, both detection or eavesdropping attacks may coexist, especially in some critical communication scenarios consisting of multiple groups with common or conflicting interests, like military communications and coastal surveillance. Therefore, in this paper we are motivated to explore a new secure wireless communication paradigm where the PLS technology is applied to counteract both the detection and eavesdropping attacks. To the best of our knowledge, this is the first paper that studies the joint guarantee for the critical covertness and secrecy properties of wireless communications at the physical layer. The main contributions of this paper are summarized as follows.

\begin{itemize}
\item \textbf{\emph{A new secure wireless communication paradigm}:} In this paradigm, the PLS technology is applied to counteract both the detection and eavesdropping attacks and thus to jointly guarantee the covertness and secrecy properties of wireless communications. To demonstrate the new paradigm, we consider four representative communication scenarios of the paradigm, which are categorized by the detector-eavesdropper relationships (i.e., \emph{independence} and \emph{friend}) and the secure transmission schemes adopted by the transmitters (i.e., a \emph{power control (PC)-based} scheme and an \emph{AN-based} scheme). In the friend relationship case, the detector group and eavesdropper group share their signals received from the target transmitters in the hope of enhancing the attack performance of both sides, while in the independence relationship case, the two groups independently conduct their own attack without such signal sharing.

\item \textbf{\emph{Theoretical modeling for the new paradigm}:} To depict the covertness, secrecy and transmission performances of the new paradigm, for each concerned communication scenario we provide the corresponding theoretical modeling of covertness outage probability (COP) (i.e., the probability that detectors detect the transmitted signals), the secrecy outage probability (SOP) (i.e., the probability that eavesdroppers recover the conveyed information) and the transmission probability (TP) (i.e., the probability of conducting transmissions), respectively.

\item \textbf{\emph{A novel security metric characterizing the covertness, secrecy and transmission performances}:} This paper defines a novel security metric-\emph{covert secrecy rate} (CSR), which characterizes the maximum transmission rate subject to the constraints of COP, SOP and TP, and thus can serve as the fundamental security criterion for this new communication paradigm. We further conduct detailed theoretical analysis to identify the CSR for each of the four communication scenarios. Finally, extensive numerical results are provided to illustrate the CSR performances under the new secure communication paradigm.
\end{itemize}

The rest of this paper is organized as follows. Section \ref{2} presents an example system for the new paradigm and the definition of CSR. Theoretical analyses for the CSR performance under the four scenarios are given in Section \ref{3} and Section \ref{4}, respectively. Section \ref{5} provides numerical results to illustrate the CSR performances and Section \ref{6} concludes this paper.

\section{New Paradigm and Security Metric} \label{2}

To demonstrate the new secure wireless communication paradigm, we consider a system (as illustrated in Fig. \ref{transmission_scheme}) where a transmitter Alice sends messages to a receiver Bob in the presence of a detector Willie and  an eavesdropper Eve. Willie attempts to detect the existence of the signals transmitted from Alice, while Eve targets the messages contained in the signals. Alice and Bob operate in the half-duplex mode, while Willie and Eve can operate in the full-duplex mode. All nodes are assumed to be equipped with a single omnidirectional antenna. For notation simplicity, we use $a$, $b$, $e$ and $w$ to represent Alice, Bob, Eve and Willie, respectively, throughout this paper.

Time is divided into successive slots with the same duration that is long enough for Alice to transmit multiple symbols. To characterize the channels, we adopt the quasi-static Rayleigh fading channel model, where the channel coefficients remain constant in one slot and change independently from one slot to another at random. We use $h_{ij}$ to denote the coefficient of the channel from $i$ to $j$, where $i\in\{a, b, e, w\}$ and $j\in\{a, b, e, w\}$. As assumed in \cite{hu2018covert}, the corresponding channel gain $|h_{ij}|^2$ follows the exponential distribution with unit mean.
We assume that Alice and Bob know the \emph{instantaneous} and \emph{statistical} channel coefficient $h_{ab}$ but only the \emph{statistical} coefficients of other channels including those to Eve and Willie. We also assume that Eve knows the \emph{instantaneous} channel coefficient $h_{ae}$, while Willie knows only the \emph{statistical} channel coefficient of $h_{aw}$ and $h_{ew}$. These assumptions are widely used in previous research related to PLS and covert communication.

\subsection{Secure Transmission Schemes}\label{sec:sec-tran}
Alice employs two transmission schemes based on power control (PC) and artificial noise (AN), respectively. In the \emph{PC-based scheme}, Alice controls her transmit power $P_a$ in order to hide the message signals into the background noise to achieve covertness and secrecy. In the \emph{AN-based scheme}, Alice intentionally injects AN into the message signals to confuse Willie and Eve so as to reduce their attack effects. Different from the PC-based scheme, in the AN-based scheme, Alice uses a constant transmit power (also denoted by $P_a$) and splits the power between message and noise transmissions. We use $\rho\in(0,1]$ to denote the fraction of transmit power used for the message transmission. In addition to the strategies of transmit power, Alice also adopts the Wyner encoding scheme \cite{wyner1975wire} to resist the eavesdropping of Eve. To transmit a message, Alice chooses a target secrecy rate $R_s$ for this message and another rate $R_t$ for the whole transmitted symbol. The difference $R_t-R_s$ represents the rate sacrificed to confuse Eve. 

The goal of Alice is to ensure a positive and \emph{constant} secrecy rate $R_s$. Thus, Alice will send messages to Bob only when the instantaneous capacity $C_b$ of the Alice-Bob channel can support the secrecy rate $R_s$ (i.e., $C_b\geq R_s$). In this situation, Alice will set $R_t$ arbitrarily close to $C_b$ to cause as much confusion to Eve as possible, while ensuring reliable message transmission to Bob. Thus, the probability of Alice transmitting messages in a certain time slot can be defined as 
\begin{equation}\label{pto}
p_{tx} = \mathbb{P}(C_b \geq R_s).
\end{equation}
Note that the \textbf{\emph{transmission probability} (TP)} $p_{tx}$ can be interpreted as a metric to measure the transmission performance.

\subsection{Attacking Model}\label{sec:att-model}
In practice, Willie and Eve can belong to different organizations with unrelated or common goals, resulting in various relationships between them. In this paper, we consider two representative relationships, i.e., \emph{independence} and \emph{friend}. As shown in Fig. \ref{transmission_scheme}, in the independence relationship, Eve and Willie care only about their own attack without helping or hindering the other. In the friend relationship, Willie and Eve will share their signals received from Alice to help improve the attack power of the other.

\begin{figure}[!t]
\centering
\includegraphics[width=0.4\textwidth]{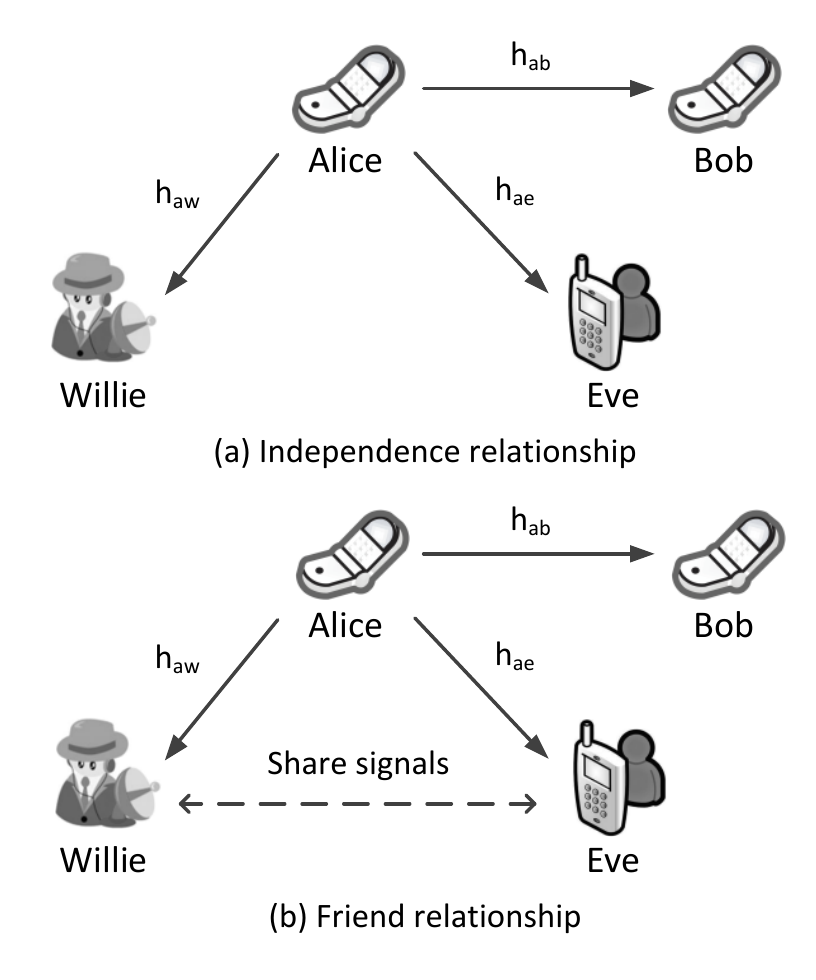}
\caption{Two relationships between Willie and Eve.}
\label{transmission_scheme}
\end{figure}

To detect the existence of signals transmitted from Alice in each slot, Willie adopts the commonly-used likelihood ratio test \cite{sobers2017covert}, in which he first determines a threshold $\theta$ and then measures the average power $\bar{P}_w$ of the symbols received from Alice in this slot. If $\bar{P}_w \ge \theta$, Willie accepts a hypothesis $\mathcal{H}_1$ that Alice transmitted messages to Bob in this slot. If $\bar{P}_w \le \theta$, Willie accepts a hypothesis $\mathcal{H}_0$ that Alice did not transmit messages. Formally, the likelihood ratio test can be given by
\begin{equation}
\bar{P}_w \mathop{\gtrless}\limits^{\mathcal{H}_1}_{\mathcal{H}_0} \theta.
\end{equation}
In general, the likelihood test introduces two types of detection errors. One is called \emph{false alarm}, which means that Willie reports a detected transmission whilst the transmission does not exist in fact. The other is called \emph{missed detection}, which means that Willie reports no detected transmission whilst the transmission exists indeed. We use $p_{FA}$ and $p_{MD}$ to denote the probabilities of false alarm and missed detection, respectively. If neither false alarm nor missed detection occurs, the transmission from Alice to Bob is said to suffer from covertness outage.  Thus, the \textbf{\emph{covertness outage probability} (COP)} is given by
\begin{equation}
p_{co} = 1-(p_{FA} + p_{MD}).
\end{equation}
The smaller the COP is, the higher the covertness of the transmission is. Note that $1-p_{co}$ can be interpreted as the detection error probability of Willie.

Compared with the detection of Willie, the eavesdropping attack of Eve is relatively simpler. To intercept the transmitted messages, Eve tries to decode the signals received from Alice. If Eve is able to recover the messages (i.e., the instantaneous secrecy capacity $C_s$ \cite{cover2012elements} of the Alice-Bob channel falls below the target secrecy rate $R_s$), the transmission from Alice to Bob is said to suffer from secrecy outage.
Note that secrecy outage occurs only when Alice actually transmits a message (i.e., $C_b\geq R_s$). Thus, we can define the \textbf{\emph{secrecy outage probability} (SOP)} as the following conditional probability: 
\begin{equation}\label{pso}
p_{so} = \mathbb{P}(C_s < R_s ~|~ C_b\geq R_s).
\end{equation}
Similarly, the smaller the SOP is, the stronger the secrecy of the transmission is.

\subsection{Covert Secrecy Rate}\label{sec:csr}
To understand the fundamental security performance under the new paradigm, we propose a novel metric, called \textbf{\emph{covert secrecy rate} (CSR)}, by jointly considering the covertness, secrecy and transmission performances. The CSR is defined as the maximum transmission rate under which the constraints of COP, SOP and TP can be ensured. To obtain the CSR, we formulate two optimization problems for the PC-based and AN-based transmission schemes, respectively, which are given by
\begin{subequations}\label{PC-based}
\begin{align}
\textbf{P1 (PC-based):}\quad R_{cs}=  \max\limits_{P_a \text{,} R_s} ~ & R_s p_{tx}(P_a,R_s) , \label{problem:csr-pc}\\
\text{s.t.}~ & ~ p_{co} (P_a)\le  \epsilon_c, \label{con:cop-pc}\\
& ~ p_{so}(R_s) \le \epsilon_s, \label{con:sop-pc}\\
& ~ p_{tx}(P_a,R_s) \ge 1-\epsilon_t, \label{con:top-pc}
\end{align}
\end{subequations}
and
\begin{subequations}\label{AN-based}
\begin{align}
\textbf{P2 (AN-based):} \ \  R_{cs} = \max \limits_{\rho\in[0,1],R_s}  & R_s p_{tx}(\rho,R_s),  \label{problem:csr-an}\\
\text{s.t.}~~ & p_{co} (\rho)\le  \epsilon_c, \label{con:cop-an}\\
&  p_{so} (\rho,R_s)\le  \epsilon_s, \label{con:sop-an}\\
&  p_{tx}(\rho,R_s) \ge 1-\epsilon_t, \label{con:top-an}
\end{align}
\end{subequations}
where $R_{cs}$ denotes the CSR, $\epsilon_c$, $\epsilon_s$ and $\epsilon_t$ denote the constraints of COP, SOP and TP. Note that Problem P1 optimizes the transmission rate over the transmit power $P_a$ and the secrecy rate $R_s$,
while Problem P2 conducts the optimization over the power allocation parameter $\rho$ and the secrecy rate $R_s$.

\section{CSR Analysis: Independence Relationship Case}\label{3}
In this section, we investigate the CSR performance under the independence relationship case, for which we focus on the PC-based and AN-based transmission schemes in Subsections \ref{sec:ind-pc} and \ref{sec:ind-an}, respectively.

\subsection{PC-Based Transmission Scheme}\label{sec:ind-pc}
As mentioned in Section \ref{sec:sec-tran}, Alice decides to transmit in a certain time slot only when the instantaneous capacity $C_b$ of Alice-Bob channel can support the secrecy rate $R_s$. To do this, Alice measures the instantaneous channel coefficient $|h_{ab}|^2$ and determines the Alice-Bob channel capacity $C_b$ based on the well-known Shannon Capacity formula \cite{cover2012elements}, i.e.,
\begin{equation}\label{eqn:cb-ind-pc}
C_b = \log \left(1+\frac{P_a |h_{ab}|^2}{ \sigma_{b}^2}\right),
\end{equation}
where $\log$ is to the base of $2$. 
Since $|h_{ab}|^2$ is exponentially distributed, the transmission probability $p_{tx}$ of Alice under the PC-based transmission scheme is 
\begin{equation}\label{eqn:pto-ind-pc}
p_{tx}^{\mathrm{IP}}(P_a,\!R_s) \!=\! \mathbb{P}\! \left(C_b \ge R_s \right) \!\!=\!\exp\! \left(\! - \frac{(2^{R_s}\!\!-\!1) \sigma_{b}^2}{P_a}\!\right).
\end{equation}

When Alice chooses to transmit, she sends $n$ symbols to Bob, represented by a complex vector $\mathbf x$, where each symbol $\mathbf x[i]$ ($i=1,2,\cdots,n$) is subject to the unit power constraint, i.e., $\mathbb{E}[\lvert \mathbf x[i] \rvert^2] = 1$.
Thus, the signal vectors received at Bob, Willie and Eve are given by
\begin{equation}\label{y_bwe}
\mathbf y_{\kappa} =\sqrt{P_a}h_{a\kappa} \mathbf x + \mathbf n_\kappa,
\end{equation}
where the subscript $\kappa\in\{b, w, e\}$ stands for Bob, Willie or Eve, $a$ represents Alice, and $\mathbf n_\kappa$ denotes the noise at $\kappa$ with the $i$-th element $\mathbf n_\kappa[i]$ being  the complex additive Gaussian noise with zero mean and variance $\sigma_{\kappa}^2$, i.e., $\mathbf n_\kappa[i] \sim \mathcal{CN}(0,\sigma_{\kappa}^2)$.

According to the detection scheme in Subsection \ref{sec:att-model}, Willie makes a decision on the existence of transmitted signals based on the average power $\bar{P}_w$ of the received symbols $\mathbf y_{w}$.
In this case, $\bar{P}_w$ is given by
\begin{align}\label{pw_ind_pc}
\bar{P}_w &= \frac{\sum_{i = 1}^{n} \lvert \mathbf y_w[i] \rvert^2}{n} = \lim_{n \to \infty}(P_a |h_{aw}|^2  + \sigma_{w}^2)\chi_{2n}^2/n \nonumber \\
&= P_a |h_{aw}|^2  + \sigma_{w}^2,
\end{align}
where $\chi_{2n}^2$ is a chi-squared random variable with $2n$ degrees of freedom.
By the Strong Law of Large Numbers \cite{browder2012mathematical}, $\frac{\chi_{2n}^2}{n}$ converges in probability to $1$ as $n$ tends to infinity.
If $\bar{P}_w \le \theta$, Willie accepts the hypothesis $\mathcal{H}_0$ that Alice did not transmit messages, leading to a missed detection.
Thus, the probability of missed detection $p_{MD}$ is given by
\begin{align}\label{eqn:pmd-ind-pc}
p_{MD} & = \mathbb{P} \left( P_a |h_{aw}|^2  + \sigma_{w}^2 \le \theta \right) \nonumber \\
& =
\begin{cases}
1- \exp \left( - \frac{\theta - \sigma_{w}^2}{P_a} \right), & \theta > \sigma_{w}^2, \\
0, & \theta \le \sigma_{w}^2.
\end{cases}
\end{align}

The eavesdropping result of Eve depends on the instantaneous secrecy capacity $C_s$ of the Alice-Bob channel, which is the difference between the channel capacity of the Alice-Bob channel and that of the Alice-Eve channel \cite{cover2012elements}. Thus, $C_s$ is formulated as
\begin{equation}
C_s = \log \left(1+\frac{P_a |h_{ab}|^2}{ \sigma_{b}^2}\right)-\log\left(1+\frac{P_a |h_{ae}|^2}{\sigma_{e}^2}\right).
\end{equation}
Note that $|h_{ab}|^2$ and $|h_{ae}|^2$ are random variables here. Based on the definition of the SOP in Subsection \ref{sec:att-model}, the SOP  under the PC-based scheme can be given by
\begin{align}\label{eqn:sop-ind-pc}
p_{so}^{\mathrm{IP}}\!(\!R_s\!) &\!\!= \frac{\mathbb{P} \left(R_s<C_b<C_e+R_s\right)}{\mathbb{P} \left( C_b > R_s\right)} = 1 - \frac{\mathbb{P} \left(C_s>R_s\right)}{\mathbb{P} \left( C_b > R_s\right)} \nonumber \\
& \!\!=\!\!1 \!\!-\! e^ {\frac{(\!2^{R_s}-1\!) \sigma_{b}^2}{P_a} }\!\mathbb{P}\! \left(\!\frac{P_a\! |h_{ab}|^2}{ \sigma_{b}^2}\!-\! \frac{2^{R_s}\!\! P_a\! |h_{ae}|^2}{ \sigma_{e}^2} \!>\! 2^{R_s}\!\!\!-\!1\!\right) \nonumber \\
&\!\!= \frac{2^{R_s} \sigma_{b}^2}{2^{R_s} \sigma_{b}^2 +\sigma_{e}^2}.
\end{align}

When Alice does not transmit, security performance is not a concern and thus we only focus on the covertness performance. In this case, Willie receives only noise, i.e., $\mathbf y_{w}=\mathbf n_w$ and thus the the average power $\bar{P}_w$ of the received symbols $\mathbf y_{w}$ is $\bar{P}_w =\sigma_w^2$. If $\bar{P}_w \ge \theta$, Willie accepts the hypothesis $\mathcal{H}_1$ that Alice transmitted messages, leading to a false alarm.
Thus, the probability of false alarm $p_{FA}$ is given by
\begin{align}\label{eqn:pfa-ind-pc}
p_{FA}  = \mathbb{P} \left(\sigma_{w}^2 \ge \theta \right) =
\begin{cases}
0, & \theta > \sigma_{w}^2, \\
1, & \theta \le \sigma_{w}^2.
\end{cases}
\end{align}

Combining the $p_{MD}$ in (\ref{eqn:pmd-ind-pc}) and the $p_{FA}$ in (\ref{eqn:pfa-ind-pc}), we obtain the COP under the PC-based scheme as
\begin{equation} \label{eqn:cop-ind-pc}
p_{co}^{\mathrm{IP}}(P_a, \theta)=
\begin{cases}
\exp \left( - \frac{\theta - \sigma_{w}^2}{P_a} \right), & \theta > \sigma_{w}^2, \\
0, & \theta \le \sigma_{w}^2.
\end{cases}
\end{equation}
Note that the COP is identical for Alice and Willie, since they have the same knowledge about $|h_{aw}|^2$, i.e., the statistical $|h_{aw}|^2$.
To maximize the COP $p_{co}^{\mathrm{IP}}$, Willie will choose the optimal detection threshold $\theta$, denoted by $\theta^*_{\mathrm{IP}}$.
We can see from (\ref{eqn:cop-ind-pc}) that $p_{co}^{\mathrm{IP}}$ is a decreasing function of $\theta$ and is larger than or equal to $0$ for $\theta > \sigma_{w}^2$.
Thus, the optimal $\theta^*_{\mathrm{IP}}$ exists in $(\sigma_{w}^2, \infty)$ and is thus given by $\theta^*_{\mathrm{IP}}=\upsilon+\sigma_{w}^2$, where $\upsilon>0$ is an arbitrarily small value.

Under the condition that Willie chooses the optimal detection threshold $\theta^*_{\mathrm{IP}}$, Alice solves the optimization problem in (\ref{PC-based}) to obtain the CSR.
The main result is summarized in the following theorem.

\begin{theorem}\label{theorem:csr-ind-pc}
Under the scenario where Willie and Eve are in the independence relationship and Alice adopts the PC-based secure transmission scheme, the CSR of the system can be given by \eqref{csr_ind_pc},
\begin{figure*}[t]	 
\begin{equation}\label{csr_ind_pc}
R_{cs}^{\mathrm{IP}} =
\begin{cases}
\frac{1}{\ln 2} \mathrm{W}_0  \left(- \frac{\upsilon}{\sigma_b^2\ln \epsilon_c}\right) \exp \left(-\frac{1}{\mathrm{W}_0  \left(- \frac{\upsilon}{\sigma_b^2\ln \epsilon_c}\right)}- \frac{\sigma_b^2\ln \epsilon_c}{\upsilon}\right), & 
R^*_{s,\mathrm{IP}}=R_{s,\mathrm{IP}}^0 \le \min\left\{R^{\text{SOP}}_{s,\mathrm{IP}},R^{\text{TP}}_{s,\mathrm{IP}}\right\}, \\
\log\left(\frac{\sigma_e^2 \epsilon_s}{(1- \epsilon_s)\sigma_b^2}\right) \exp\left( \frac{\left(\sigma_e^2 \epsilon_s-(1-\epsilon_s)\sigma_b^2\right)\ln \epsilon_c}{(1-\epsilon_s)\upsilon}\right), & R^*_{s,\mathrm{IP}}=R^{\text{SOP}}_{s,\mathrm{IP}} \le \min\left\{R_{s,\mathrm{IP}}^0,R^{\text{TP}}_{s,\mathrm{IP}}\right\}, \\
(1-\epsilon_t)\log\left(1+\frac{\upsilon \ln(1-\epsilon_t)}{\sigma_b^2 \ln \epsilon_c}\right)
, & R^*_{s,\mathrm{IP}}=R^{\text{TP}}_{s,\mathrm{IP}} \le \min\left\{R_{s,\mathrm{IP}}^0,R^{\text{SOP}}_{s,\mathrm{IP}}\right\},
\end{cases}
\end{equation}
{\noindent} 
\rule[-10pt]{18.07cm}{0.05em}
\end{figure*}
where
\begin{equation}\label{IP_Rs_sop}
R^{\text{SOP}}_{s,\mathrm{IP}} = \log\left(\frac{\sigma_e^2 \epsilon_s}{(1-\epsilon_s)\sigma_b^2}\right),
\end{equation}
\begin{equation}\label{IP_Rs_top}
R^{\text{TP}}_{s,\mathrm{IP}}  = \log\left(1-\frac{P_{a,\mathrm{IP}}^*\ln(1-\epsilon_t)}{\sigma_b^2}\right),
\end{equation}
\begin{equation}\label{IP_Rs_Rcs}
R_{s,\mathrm{IP}}^0 = \frac{1}{\ln 2} \mathrm{W}_0 \! \left( \frac{P_{a,\mathrm{IP}}^*}{\sigma_b^2}\right),
\end{equation}
$\mathrm{W}_0 (\cdot)$ is the principal branch of Lambert's W function,
and $P_{a,\mathrm{IP}}^*=-\frac{\upsilon}{\ln \epsilon_c}$ is the optimal transmit power.
\end{theorem}

\begin{proof}
As can be seen from \eqref{problem:csr-pc}, the optimal transmit power $P_a$ and optimal target secrecy rate $R_s$ are required to solve the optimization problem P1. We first derive the optimal $P_a$. 
It is easy to see from \eqref{eqn:pto-ind-pc} and \eqref{eqn:cop-ind-pc} that both $p_{tx}^{\mathrm{IP}}$ and $p_{co}^{\mathrm{IP}}$ monotonically increase as $P_a$ increases. Thus, the covertness constraint in (\ref{con:cop-pc}) results in an upper bound on $P_a$, which is
\begin{equation}
P_{a,\mathrm{IP}}^{\max}=-\frac{\upsilon}{\ln \epsilon_c},
\end{equation}
and the TP constraint in (\ref{con:top-pc}) leads to a lower bound on $P_a$, which is
\begin{equation}
P_{a,\mathrm{IP}}^{\min}=-\frac{(2^{R_s}-1)\sigma_b^2}{\ln (1-\epsilon_t)}.
\end{equation}
Note that the inequality $P_{a,\mathrm{IP}}^{\min}\le P_{a,\mathrm{IP}}^{\max}$ must hold, which gives the following condition on $R_s$:
\begin{equation}\label{IP_Rs_pa}
R_s \leq \log\left(1+\frac{\upsilon \ln(1-\epsilon_t)}{\sigma_b^2 \ln \epsilon_c}\right).
\end{equation}
Since the objective function in (\ref{problem:csr-pc}) is an increasing function of $P_a$, the optimal $P_a$ is the upper bound, i.e., $P_{a,\mathrm{IP}}^*=P_{a,\mathrm{IP}}^{\max}$.


Next, we derive the optimal $R_s$ by analyzing the feasible region of $R_s$ and the monotonicity of the objective function with respect to $R_s$. We can see that as $R_s$ increases, $p_{tx}^{\mathrm{IP}}$ in \eqref{eqn:pto-ind-pc} monotonically decreases while $p_{so}^{\mathrm{IP}}$ in \eqref{eqn:sop-ind-pc} monotonically increases. Thus, based on the constraints (\ref{con:sop-pc}) and (\ref{con:top-pc}), the regions of $R_s$ for ensuring secrecy and transmission performances are $[0,R^{\text{SOP}}_{s,\mathrm{IP}}]$ and $[0,R^{\text{TP}}_{s,\mathrm{IP}} ]$ with $R^{\text{SOP}}_{s,\mathrm{IP}}$ and $R^{\text{TP}}_{s,\mathrm{IP}} $ given by (\ref{IP_Rs_sop}) and (\ref{IP_Rs_top}), respectively. Note that $R^{\text{TP}}_{s,\mathrm{IP}} $ is obtained at $P_a=P_{a,\mathrm{IP}}^*=-\frac{\upsilon}{\ln \epsilon_c}$ and thus the region  $[0,R^{\text{TP}}_{s,\mathrm{IP}} ]$ is equivalent to  \eqref{IP_Rs_pa}. Hence, the feasible region of $R_s$ is $[0,\min\{R^{\text{SOP}}_{s,\mathrm{IP}}, R^{\text{TP}}_{s,\mathrm{IP}} \}]$.
Taking the first derivative of the objective function in (\ref{problem:csr-pc}) in terms of $R_s$ gives
\begin{equation}\label{IP_Rcs_par_Rs}
\frac{\partial R_{cs}}{\partial R_s} = \left( 1\!-\!\frac{R_s 2^{R_s} \sigma_b^2 \ln 2}{P_a}\right) \exp\left( -\frac{(2^{R_s}\!-\!1)\sigma_b^2}{P_a}\right).
\end{equation}
Solving $\frac{\partial R_{cs}}{\partial R_s}=0$, we can obtain the stationary point $R_{s,\mathrm{IP}}^0$ in (\ref{IP_Rs_Rcs}).
We can see that the objective function is increasing over $[0,R_{s,\mathrm{IP}}^0 )$ and decreasing over $[R_{s,\mathrm{IP}}^0,\infty)$. 
This implies that if $R_{s,\mathrm{IP}}^0$ falls inside the feasible region of $R_s$, i.e., $R_{s,\mathrm{IP}}^0\leq \min\{R^{\text{SOP}}_{s,\mathrm{IP}}, R^{\text{TP}}_{s,\mathrm{IP}} \}$, the optimal $R_s$ is $R^*_{s,\mathrm{IP}}=R_{s,\mathrm{IP}}^0$. Otherwise, the optimal $R_s$ is $R^*_{s,\mathrm{IP}}=\min\{R^{\text{SOP}}_{s,\mathrm{IP}}, R^{\text{TP}}_{s,\mathrm{IP}} \}$. 
Finally, substituting the optimal $P_a$ and $R_s$ into the objective function in (\ref{problem:csr-pc}) completes the proof.
\end{proof}

\subsection{AN-Based Transmission Scheme} \label{sec:ind-an}
Suppose Alice transmits, in addition to the message symbols, she will also inject AN, represented by a complex vector $\mathbf z$, where each symbol $\mathbf z[i]$ ($i=1,2,\cdots,n$) is subject to the unit power constraint, i.e., $\mathbb{E}[\lvert \mathbf z[i] \rvert^2] = 1$.
Alice will use a fraction $\rho$ of her transmit power $P_a$ for message transmission and the remaining power for AN radiation.
Thus, the signal vectors received at Bob will be given by
\begin{equation}\label{eqn:yb-ind-an}
\mathbf y_{b} = \sqrt{\rho P_a}h_{ab}\mathbf x +\sqrt{(1-\rho) P_a}h_{ab}\mathbf z + \mathbf n_b.
\end{equation}

Based on (\ref{eqn:yb-ind-an}), Alice measures the instantaneous Alice-Bob channel capacity $C_b$ as
\begin{equation}\label{eqn:cb-ind-an}
C_b = \log \left(1+\frac{\rho P_a |h_{ab}|^2}{(1-\rho) P_a |h_{ab}|^2+ \sigma_{b}^2}\right),
\end{equation}
and decides to transmit when $C_b\ge R_s$.
Thus, the transmission probability under the AN-based scheme can be given by
\begin{align}\label{eqn:pto-ind-an}
p_{tx}^{\mathrm{IA}}(\rho,R_s) &=\mathbb{P} \left( C_b\geq R_s\right)\nonumber\\
&= \mathbb{P} \left( \frac{\rho P_a |h_{ab}|^2}{(1-\rho) P_a |h_{ab}|^2+ \sigma_{b}^2} \ge 2^{R_s}-1 \right) \nonumber \\
&= \exp \left(-\frac{(2^{R_s}-1) \sigma_{b}^2}{\rho P_a-(2^{R_s}-1)(1-\rho)P_a}\right).
\end{align}

Next, we analyze the secrecy and covertness performances when Alice transmits messages. In this situation, the signal vectors received at Willie and Eve have the same form of that received at Bob, which are given by
\begin{equation}\label{eqn:ybwe-ind-an}
\mathbf y_{\kappa} = \sqrt{\rho P_a}h_{a\kappa}\mathbf x +\sqrt{(1-\rho) P_a}h_{a\kappa}\mathbf z + \mathbf n_\kappa,
\end{equation}
where the subscript $\kappa\in\{w, e\}$ stands for Willie or Eve. From (\ref{eqn:ybwe-ind-an}), we can see that the average power $\bar{P}_w$ of the received symbols $\mathbf y_{\kappa}$ at Willie is the same as that given in (\ref{pw_ind_pc}).
Thus, the probability of missed detection $p_{MD}$ under the AN-based scheme can also be given by (\ref{eqn:pmd-ind-pc}).

According to (\ref{eqn:ybwe-ind-an}), the secrecy capacity $C_s$ under the AN-based scheme can be formulated as
\begin{equation}
C_s \!=\! \log\!\! \left(\!\!1\!\!+\!\!\frac{\rho P_a |h_{ab}|^2}{(\!1\!-\!\rho\!)\! P_a\! |h_{ab}|^2\!\!+\!\! \sigma_{b}^2}\!\right)\!\!-\!\!\log\!\left(\!\!1\!\!+\!\!\frac{\rho P_a |h_{ae}|^2}{(\!1\!-\!\rho\!) \!P_a\! |h_{ae}|^2\!+\!\sigma_{e}^2}\!\right)\!.
\end{equation}
Thus, following the definition of SOP in (\ref{pso}), we derive the SOP under the AN-based scheme as
\begin{align}\label{eqn:sop-ind-an}
&p_{so}^{\mathrm{IA}}(\rho,\!R_s)  \!\!=\!\! 1- \exp \left( \frac{(2^{R_s}-1) \sigma_{b}^2}{\rho P_a-(2^{R_s}-1)(1-\rho)P_a}\right) \\
&\!\!\!\!\ \ \ \times\!\mathbb{P} \bigg(\bigg.\! \!\frac{\rho P_a |h_{ab}|^2}{(\!1\!\!-\!\!\rho\!)\! P_a\! |h_{ab}|^2\!\!+\!\! \sigma_{b}^2} \!\!-\!\! \frac{2^{R_s} \rho P_a |h_{ae}|^2}{(\!1\!\!-\!\!\rho\!)\! P_a\! |h_{ae}|^2\!\!+\!\!\sigma_{e}^2}\!\!>\!\!2^{R_s}\!\!-\!\!1 \!\!\bigg.\bigg) \nonumber \\
& \!\!=\!\!1\!-\! \exp\!\left(\! \frac{(2^{R_s}-1) \sigma_{b}^2}{\rho P_a\!\!-\!\!(\!2^{R_s}\!-\!1\!)(\!1\!\!-\!\!\rho\!)\!P_a}\!-\!\frac{(2^{R_s}\!+\!\rho\!-\!1)\sigma_{b}^2}{(\!1\!-\!2^{R_s}\!)(\!1\!\!-\!\!\rho\!)\!P_a}\!\right)\nonumber \\
&\!\!\!\! \ \ \ \int_{0}^{\phi}\!\exp\!\! \left(\!\frac{\frac{(\!2^{R_s}\!+\!\rho-\!1\!)(\!1-(\!1\!-\!\rho\!)2^{R_s}\!)\sigma_{b}^2\sigma_{e}^2}{(1-2^{R_s})(1\!-\!\rho)}\!-\!(\!2^{R_s}\!\!-\!\!1\!)\sigma_{b}^2\sigma_{e}^2}{(\!1\!-\!2^{R_s}\!)(\!1\!\!-\!\!\rho\!)P_a^2 y\!+\! (\!1\!-\!(\!1\!\!-\!\!\rho\!)2^{R_s}\!) P_a\sigma_{e}^2}\!-\!y\right) \mathrm{d}y, \nonumber
\end{align}
where $\phi=\frac{(1- 2^{R_s}(1-\rho))\sigma_{e}^2}{(2^{R_s}-1)(1-\rho)P_a}$.

Finally, we analyze the covertness performance when Alice does not transmit messages. In this situation, Alice still generates AN to confuse Willie, which is different from the PC-based scheme.
Thus, the signal vector $\mathbf y_{w}$ received by Willie consists of both the AN $\mathbf z$ and background noise, i.e.,
\begin{equation}\label{eqn:yW-half-ind-an}
\mathbf y_{w}=\sqrt{(1-\rho)P_a}h_{aw}\mathbf z +\mathbf n_w.
\end{equation}
In this case, the average power of the received symbols of Willie is $\bar{P}_w =(1-\rho)P_a|h_{aw}|^2+ \sigma_w^2$, and thus the probability of false alarm is given by
\begin{align}\label{eqn:pfa-ind-an}
p_{FA}  &= \mathbb{P} \left((1-\rho)P_a|h_{aw}|^2+ \sigma_w^2 \ge \theta \right) \nonumber \\
&=
\begin{cases}
\exp \left( - \frac{ (\theta - \sigma_{w}^2)}{(1 -\rho)P_a} \right), & \theta > \sigma_{w}^2, \\
1, & \theta  \le \sigma_{w}^2.
\end{cases}
\end{align}

Combining the $p_{FA}$ in (\ref{eqn:pfa-ind-an}) and the $p_{MD}$ in (\ref{eqn:pmd-ind-pc}), we obtain the COP $p_{co}^{\mathrm{IA}}$ under the AN-based scheme as
\begin{equation} \label{eqn:cop-ind-an}
p_{co}^{\mathrm{IA}}(\rho, \theta) \!=\!
\begin{cases}
\!\exp\! \left(\!-\frac{ (\theta-\sigma_{w}^2)}{P_a}\! \right)\!-\! \exp\! \left(\!-\frac{ (\theta - \sigma_{w}^2)}{(1\!-\!\rho)P_a}\!\right),\!\!\! &\theta \!>\! \sigma_{w}^2, \\
\!0,\!\!\! & \theta \!\le\! \sigma_{w}^2.
\end{cases}
\end{equation}

We can see from (\ref{eqn:cop-ind-an}) that the optimal detection threshold $\theta^*_{\mathrm{IA}}$ for Willie exists when $\theta>\sigma_{w}^2$ and can be obtained by solving $\frac{\partial p_{co}^{\mathrm{IA}}}{\partial \theta}=0$.
Thus, $\theta^*_{\mathrm{IA}}$ is given by
\begin{equation}\label{theta_ind_an}
\theta^*_{\mathrm{IA}} = \sigma_{w}^2+ \frac{(\rho -1)P_a}{\rho}\ln(1-\rho).
\end{equation}

By solving the optimization problem in (\ref{AN-based}) with  $\theta=\theta^*_{\mathrm{IA}}$, we can obtain the CSR, which is given in the following theorem.
\begin{theorem}\label{theorem:csr-ind-an}
Under the scenario where Willie and Eve are in the independence relationship and Alice adopts the AN-based secure transmission scheme, the CSR of the system is
\begin{equation}\label{eqn:csr-ind-an}
R_{cs}^{\mathrm{IA}} \!\!=\!\!R^*_{s,\mathrm{IA}}(\rho^*_{\mathrm{IA}})\!\exp\! \left(\!\!- \frac{(2^{R^*_{s,\mathrm{IA}}(\rho^*_{\mathrm{IA}})}-1) \sigma_{b}^2}{\rho^*_{\mathrm{IA}}\! P_a\!-\! (\!2^{R^*_{s,\mathrm{IA}}\!(\rho^*_{\mathrm{IA}})}\!-\!1\!)(\!1\!-\!\rho^*_{\mathrm{IA}}\!)P_a}\!\right),
\end{equation}
where $\rho^*_{\mathrm{IA}}$ is the optimal power allocation parameter and $R^*_{s,\mathrm{IA}}$ is the optimal secrecy rate. 
Here, $\rho^*_{\mathrm{IA}}$ can be obtained by solving $p_{co}^{\mathrm{IA}} (\rho, \theta^*_{\mathrm{IA}})= \epsilon_c$ with $\theta^*_{\mathrm{IA}}$ given by (\ref{theta_ind_an}).
$R^*_{s,\mathrm{IA}}$ is given by
\begin{align}\label{rs_opt_IA}
R^*_{s,\mathrm{IA}}\!(\rho^*_{\mathrm{IA}})\!=\!\!
\begin{cases}
\!R_{s,\mathrm{IA}}^0(\rho^*_{\mathrm{IA}}),\!\!\!\!\!
&R^*_{s,\mathrm{IA}}\!\!=\!\!R_{s,\mathrm{IA}}^0 \!\le\! \min\!\left\{\!R^{\text{SOP}}_{s,\mathrm{IA}}\!,R^{\text{TP}}_{s,\mathrm{IA}}\!\right\}\!, \\
\!R^{\text{SOP}}_{s,\mathrm{IA}}(\rho^*_{\mathrm{IA}}),\!\!\!\!\! &R^*_{s,\mathrm{IA}}\!\!=\!\!R^{\text{SOP}}_{s,\mathrm{IA}} \!\le\! \min\!\left\{\!R_{s,\mathrm{IA}}^0\!,R^{\text{TP}}_{s,\mathrm{IA}}\!\right\}\!, \\
\!R^{\text{TP}}_{s,\mathrm{IA}}(\rho^*_{\mathrm{IA}}),\!\!\!\!\! &R^*_{s,\mathrm{IA}}\!\!=\!\!R^{\text{TP}}_{s,\mathrm{IA}} \!\le\! \min\!\left\{\!R_{s,\mathrm{IA}}^0\!,R^{\text{SOP}}_{s,\mathrm{IA}}\!\right\}\!,
\end{cases}
\end{align}
where the stationary point $R_{s,\mathrm{IA}}^0$ can be obtained by solving $\frac{\partial R_{cs}}{\partial R_s}=0$, $R^{\text{SOP}}_{s,\mathrm{IA}}$ is the solution of $p_{so}^{\mathrm{IA}}(R_s)=\epsilon_s$ and $R^{\text{TP}}_{s,\mathrm{IA}}$ is given by
\begin{equation}\label{rs_top_IA}
R^{\text{TP}}_{s,\mathrm{IA}}(\rho^*_{\mathrm{IA}}) =\log\left( \frac{P_a\ln(1-\epsilon_t)- \sigma_{b}^2}{(1-\rho^*_{\mathrm{IA}}) P_a \ln(1- \epsilon_t)-\sigma_{b}^2}\right).
\end{equation}
\end{theorem}

\begin{proof}
The proof follows the same idea as the one for Theorem \ref{theorem:csr-ind-pc}. The only difference is to derive the optimal power allocation parameter $\rho$ instead of optimal transmit power $P_a$. Here, we focus on the derivation of the optimal $\rho$ and omit the analysis of the optimal $R_s$. 
We can see that the objective function in (\ref{problem:csr-an}) is an increasing function of $\rho$, implying that  the upper bound on $\rho$ is needed. 
Substituting $\theta=\theta^*_{\mathrm{IA}}$ into (\ref{eqn:cop-ind-an}) yields
\begin{equation}\label{pco_ind_re}
p_{co}^{\mathrm{IA}} = \rho(1-\rho)^{\frac{1-\rho}{\rho}}.
\end{equation}
Taking the first derivative of  (\ref{pco_ind_re}) in terms of $\rho$, we have
\begin{equation}
\frac{\partial p_{co}^{\mathrm{IA}}}{\partial \rho} = \frac{-\ln (1-\rho)}{\rho} \left(1-\rho\right)^{\frac{1-\rho}{\rho}} >0,
\end{equation}
which shows that $p_{co}^{\mathrm{IA}}$ is an increasing function of $\rho$. We can see from \eqref{eqn:pto-ind-an} and \eqref{eqn:sop-ind-an} that $p_{tx}^{\mathrm{IA}}$ is also an increasing function of $\rho$, while $\rho^{\text{SOP}}_{\mathrm{IA}}$ is a decreasing function. Thus, only the covertness constraint $\eqref{con:cop-an}$ gives an upper bound $\rho^{\max}_{\mathrm{IA}}$ on $\rho$, while the TP and SOP constraints in \eqref{con:top-an}) and \eqref{con:sop-an} give two lower bounds $\rho^{\text{TP}}_{\mathrm{IA}}$ and $\rho^{\text{SOP}}_{\mathrm{IA}}$ respectively. Hence, the optimal $\rho$ is $\rho^*_{\mathrm{IA}}=\rho^{\max}_{\mathrm{IA}}$. Note that $\rho^{\max}_{\mathrm{IA}}\ge\max \left\{\rho^{\text{TP}}_{\mathrm{IA}},\rho^{\text{SOP}}_{\mathrm{IA}}\right\}$ must hold, which imposes a constraint (or region) on $R_s$. However, this region is equivalent to the one obtained from the TP and SOP constraints in \eqref{con:top-an}) and \eqref{con:sop-an}, and thus can be neglected in the analysis of optimal $R_s$.
\end{proof}

\section{CSR Analysis: Friend Relationship Case} \label{4}
The CSR performance of the friend relationship case is investigated in this section, for which the CSR analyses for the PC-based and AN-based transmission schemes are provided in Subsections \ref{sec:fri-pc} and \ref{sec:fri-an}, respectively.
To depict the friend relationship, we interpret Willie and Eve as two antennas of a super attacker.
This model is widely used to characterize the collusion among eavesdroppers \cite{cumanan2013secrecy}.

\subsection{PC-Based Transmission Scheme}\label{sec:fri-pc}
Alice follows the same decision process as introduced in Section \ref{sec:ind-pc} to decide whether to transmit messages or not. Note that the instantaneous Alice-Bob channel capacity $C_b$ in this case is identical to that in \eqref{eqn:cb-ind-pc}, which means that the transmission probability is also the same.
Thus, the transmission probability $p_{tx}^{\mathrm{FP}}$ in the friend relationship scenario under the PC-based scheme is given by \eqref{eqn:pto-ind-pc}.

Next, we analyze the covertness and secrecy performances when Alice transmits messages.
When Alice chooses to transmit a signal vector $\mathbf x$, Willie and Eve receive the same signal vectors $\mathbf y_{w}$ and $\mathbf y_{e}$ as that given in (\ref{y_bwe}).
Since Willie and Eve share their received signals in this case, the signal vectors received at Willie and Eve contain the one from the other side.  
Thus, based on the signal vector $\mathbf y_{\kappa}$ in (\ref{y_bwe}), the average power of the received symbols at Willie can be given by $\bar{P}_w = \sum_{\kappa\in\{w, e\}}|\mathbf y_{\kappa}|^2=P_a |h_{aw}|^2 + P_a |h_{ae}|^2  + \sigma_{e}^2+ \sigma_{w}^2$.
Note that $|h_{aw}|^2$ and $|h_{ae}|^2$ are random variables for Willie.
Thus, the probability of missed detection $p_{MD}$ is given by
\begin{align}\label{pmd_fri_pc}
p_{MD} & = \mathbb{P} \left( P_a |h_{aw}|^2 + P_a |h_{ae}|^2  + \sigma_{e}^2+ \sigma_{w}^2 \le \theta \right) \\
& =
\begin{cases}
1\!-\!\frac{P_a\!+\theta-\sigma_{e}^2\!- \sigma_{w}^2}{P_a}\exp\!\left(\!- \frac{\theta-\sigma_{e}^2- \sigma_{w}^2}{P_a}\right), & \theta > \sigma_{e}^2\!+\! \sigma_{w}^2, \\
0, & \theta \le \sigma_{e}^2\!+\! \sigma_{w}^2.
\end{cases}\nonumber
\end{align}

According to \cite{zhang2015secure}, the signal sharing results in an improved Signal-to-Noise Ratio (SNR) for Eve, which is $\frac{P_a |h_{ae}|^2+P_a |h_{aw}|^2}{ \sigma_{e}^2+\sigma_{w}^2}$.
Thus, the secrecy capacity $C_s$ is
\begin{equation}\label{eqn:sc-frd}
C_s \!=\! \log \left(\!1\!+\!\frac{P_a |h_{ab}|^2}{ \sigma_{b}^2}\!\right) \!-\!\log\left(\!1\!+\!\frac{P_a |h_{ae}|^2+P_a |h_{aw}|^2}{ \sigma_{e}^2+\sigma_{w}^2}\!\right).
\end{equation}
Since $|h_{ab}|^2$, $|h_{ae}|^2$ and $|h_{aw}|^2$ are independent, the SOP under the PC-based scheme is given by
\begin{align}\label{pso_fri_pc}
p_{so}^{\mathrm{FP}}(R_s) &=1 - \exp \left( \frac{(2^{R_s}-1) \sigma_{b}^2}{P_a} \right) \nonumber \\
&\times\mathbb{P}\!\left(\!\frac{P_a |h_{ab}|^2}{\sigma_{b}^2} \!-\! 2^{R_s}\frac{P_a |h_{aw}|^2\!+\!P_a |h_{ae}|^2}{ \sigma_{w}^2\!+\!\sigma_{e}^2} \!>\!2^{R_s}\!-\!1\!\right) \nonumber \\
&= \frac{2^{R_s}\sigma_{b}^2(2^{R_s}\sigma_{b}^2 +2\sigma_{w}^2+2\sigma_{e}^2)}{(2^{R_s}\sigma_{b}^2+\sigma_{w}^2+\sigma_{e}^2)^2}.
\end{align}

Finally, we focus on the covertness performance when Alice suspends her transmission. Since the decision of suspending transmission is \emph{unknown} to Willie and Eve, they still share their signals, which contain only background noises.
Thus, the received signal at Willie is given by $\mathbf y_{w} =\mathbf n_e +\mathbf n_w$ and the average received power  is $\bar{P}_w= \sigma_e^2+ \sigma_w^2$. Hence, the probability of false alarm $p_{FA}$ can be given by
\begin{align}\label{pfa_fri_pc}
p_{FA} = \mathbb{P} \left(\sigma_e^2+ \sigma_w^2 \ge \theta \right) =
\begin{cases}
0, & \theta > \sigma_e^2+ \sigma_w^2, \\
1, & \theta \le \sigma_e^2+ \sigma_w^2.
\end{cases}
\end{align}

Combining the $p_{FA}$ in (\ref{pfa_fri_pc}) and the $p_{MD}$ in (\ref{pmd_fri_pc}), we obtain the COP as
\begin{equation} \label{pco_fri_pc}
p_{co}^{\mathrm{FP}} (P_a,\!\theta)\!\!=\!\!
\begin{cases}
\!\frac{P_a+\theta-\sigma_{e}^2- \sigma_{w}^2}{P_a}\exp\!\left(\!- \frac{\theta-\sigma_{e}^2- \sigma_{w}^2}{P_a}\!\right), \!\!& \theta > \sigma_{e}^2+ \sigma_{w}^2,\\
\!0, \!\!& \theta \le \sigma_{e}^2+ \sigma_{w}^2.
\end{cases}
\end{equation}
Taking the derivative of the $p_{co}^{\mathrm{FP}}$ in (\ref{pco_fri_pc}) gives
\begin{align}
\frac{\partial p_{co}^{\mathrm{FP}}}{\partial \theta}=-\frac{\theta-\sigma_{e}^2- \sigma_{w}^2}{P_a^2}\exp\left(-\frac{\theta-\sigma_{e}^2- \sigma_{w}^2}{P_a}\right).
\end{align}
This shows that $p_{co}^{\mathrm{FP}}$ is a decreasing function of $\theta$ when $\theta > \sigma_{e}^2+ \sigma_{w}^2$.
Thus, the optimal detection threshold is
\begin{align}\label{eqn:opt-theta-fri-pc}
\theta^*_{\mathrm{FP}}=\upsilon+\sigma_{e}^2+ \sigma_{w}^2,
\end{align}
where $\upsilon>0$ is an arbitrarily small value.

Given the $\theta^*_{\mathrm{FP}}$, the $p_{tx}$ in (\ref{eqn:pto-ind-pc}), the SOP in (\ref{pso_fri_pc}) and the COP in (\ref{pco_fri_pc}), the problem in (\ref{PC-based}) can now be solved to obtain the CSR.
The result is given in the following theorem.
\begin{theorem}
Under the scenario where Willie and Eve are in the friend relationship and Alice adopts the PC-based secure transmission scheme, the CSR of the system is given in \eqref{csr_fri_pc},
\begin{figure*}[t]
\begin{align}\label{csr_fri_pc}
R_{cs}^{\mathrm{FP}}\!\!=\!\!
\begin{cases}
\!\frac{1}{\ln 2} \mathrm{W}_0 \! \left(\!- \frac{\upsilon}{\left(1+\mathrm{W}_{-1}\!(-\frac{\epsilon_c}{\boldsymbol{e}})\right)\sigma_b^2}\!\right) \!\exp\! \left(\!-\frac{1}{\mathrm{W}_0 \! \left(\!- \frac{\upsilon}{\left(1+\mathrm{W}_{-1}\!(-\frac{\epsilon_c}{\boldsymbol{e}})\right)\sigma_b^2}\! \right)}\!-\! \frac{\left(1+\mathrm{W}_{-1}\!(-\frac{\epsilon_c}{\boldsymbol{e}})\right)\sigma_b^2}{\upsilon}\!\right)\!,\!\!\!\! &R^*_{s,\mathrm{FP}}\!=\!R_{s,\mathrm{FP}}^0 \!\le\! \min\!\left\{R^{\text{SOP}}_{s,\mathrm{FP}},R^{\text{TP}}_{s,\mathrm{FP}}\right\}\!, \\
\!\log\!\left(\frac{(1-\sqrt{1-\epsilon_s})(\sigma_w^2\!+\!\sigma_e^2)}{\sigma_b^2\sqrt{1-\epsilon_s}}\right) \exp\!\left(\! \frac{\left((1-\sqrt{1-\epsilon_s})(\sigma_w^2\!+\!\sigma_e^2)-\sqrt{1-\epsilon_s}\sigma_b^2\right)\left(1+\mathrm{W}_{-1}\!(-\frac{\epsilon_c}{\boldsymbol{e}})\right)}{\upsilon\sqrt{1-\epsilon_s}}\!\right)\!,\!\!\!\! &R^*_{s,\mathrm{FP}}\!=\!R^{\text{SOP}}_{s,\mathrm{FP}} \!\le\! \min\!\left\{R_{s,\mathrm{FP}}^0,R^{\text{TP}}_{s,\mathrm{FP}}\right\}\!, \\
\!(1-\epsilon_t)\log\left(1+\frac{\upsilon \ln(1-\epsilon_t)}{\sigma_b^2 \left(1+\mathrm{W}_{-1}\!(-\frac{\epsilon_c}{\boldsymbol{e}})\right)}\right)\!,\!\!\!\! &R^*_{s,\mathrm{FP}}\!=\!R^{\text{TP}}_{s,\mathrm{FP}} \!\le\! \min\!\left\{R_{s,\mathrm{FP}}^0,R^{\text{SOP}}_{s,\mathrm{FP}}\right\}\!,
\end{cases}
\end{align}
{\noindent} 
\rule[-10pt]{18.07cm}{0.05em}
\end{figure*}
Here,
\begin{equation} \label{FP_Rs_sop}
R^{\text{SOP}}_{s,\mathrm{FP}} = \log\left(\frac{(1-\sqrt{1-\epsilon_s})(\sigma_w^2+\sigma_e^2)}{\sigma_b^2\sqrt{1-\epsilon_s}}\right),
\end{equation}
$R^{\text{TP}}_{s,\mathrm{FP}}$ and $R_{s,\mathrm{FP}}^0$ are the same as those given in (\ref{IP_Rs_top}) and (\ref{IP_Rs_Rcs}), respectively, with the optimal transmit power $P_{a,\mathrm{FP}}^*$ given by
\begin{equation}
P_{a,\mathrm{FP}}^*=-\frac{\upsilon}{1+\mathrm{W}_{-1}(-\frac{\epsilon_c}{\boldsymbol{e}})}.
\end{equation}
$\mathrm{W}_0 (\cdot)$ and $\mathrm{W}_{-1} (\cdot)$ are the principal branch and the non-principle branch of Lambert's W function, respectively, and $\boldsymbol{e}$ is Euler's number.
\end{theorem}

\begin{proof}
The proof is similar to that of Theorem \ref{theorem:csr-ind-pc} and thus omitted here.
\end{proof}

\subsection{AN-Based Transmission Scheme}\label{sec:fri-an}
We first derive the transmission probability to characterize the transmission performance of the transmission. 
Suppose Alice transmits under the AN-based scheme, Bob will receive the same signal as that given in \eqref{eqn:ybwe-ind-an}, yielding the same instantaneous Alice-Bob channel capacity $C_b$ as that given in \eqref{eqn:cb-ind-an}.
This means that the transmission probability $p_{tx}^{\mathrm{FA}}$ under the AN-based scheme in the friend relationship scenario is identical to that in the independence scenario, which is given in \eqref{eqn:pto-ind-an}.

We proceed to analyze the miss detection probability and SOP when Alice transmits messages.
When Alice transmits a signal vector $\mathbf x$, the signal vectors at Willie and Eve are the same as that given in \eqref{eqn:ybwe-ind-an}. After receiving the shared signals from Eve, the average power $\bar{P}_w$ of the received symbols at Willie is given by $\bar{P}_w = \sum_{\kappa\in\{w, e\}}|\mathbf y_{\kappa}|^2= P_a |h_{ae}|^2+P_a |h_{aw}|^2+\sigma_{e}^2+\sigma_{w}^2$, which is identical to \eqref{pw_ind_pc}, i.e., the average power in the independence case. Thus, the probability of missed detection $p_{MD}$  can be given by (\ref{pmd_fri_pc}).

After Eve receives the signals from Willie, the Signal-to-Noise-plus-Interference Ratio (SINR) is  
\begin{align}
\frac{\rho P_a |h_{ae}|^2\!+\!\rho P_a |h_{aw}|^2}{(\!1\!\!-\!\!\rho\!) P_a |h_{ae}|^2\!+\!(\!1\!\!-\!\!\rho\!) P_a |h_{aw}|^2\!+\!\sigma_{e}^2\!+\!\sigma_{w}^2}.
\end{align}
Thus, the secrecy capacity $C_s$ under the AN-based scheme is
\begin{align}
C_s \!=\!&\log \left(1+\frac{\rho P_a |h_{ab}|^2}{(1-\rho) P_a |h_{ab}|^2+ \sigma_{b}^2}\right)  \\
&-\log\bigg(\!1\!+\!\frac{\rho P_a |h_{ae}|^2\!+\!\rho P_a |h_{aw}|^2}{(\!1\!\!-\!\!\rho\!) P_a |h_{ae}|^2\!+\!(\!1\!\!-\!\!\rho\!) P_a |h_{aw}|^2\!+\!\sigma_{e}^2\!+\!\sigma_{w}^2}\!\bigg),\nonumber
\end{align}
According to the definition in (\ref{pso}), the SOP is given by (\ref{eqn:sop-fri-an}).
\begin{figure*}[t]
\begin{align}\label{eqn:sop-fri-an}
p_{so}^{\mathrm{FA}}(\rho,R_s) & \!=\! 1-\exp\! \left(\! \frac{(2^{R_s}-1) \sigma_{b}^2}{\rho P_a\!-\!(\!2^{R_s}\!-\!1\!)(\!1\!-\!\rho\!)P_a}\!\right)\mathbb{P} \left(\! \frac{\rho P_a |h_{ab}|^2}{(\!1\!-\!\rho\!)\! P_a |h_{ab}|^2\!+ \!\sigma_{b}^2}\!-\frac{2^{R_s}(\rho P_a |h_{ae}|^2\!+\!\rho P_a |h_{aw}|^2)}{(\!1\!-\!\rho\!)\! P_a |h_{ae}|^2\!+\!(\!1\!-\!\rho\!)\! P_a |h_{aw}|^2\!+\!\sigma_{e}^2\!+\!\sigma_{w}^2}\!>\! 2^{R_s}\!-\!1 \!\right) \nonumber  \\
& \!=\!1-\exp \left( \frac{(2^{R_s}-1) \sigma_{b}^2}{\rho P_a-(2^{R_s}-1)(1-\rho)P_a}-\frac{(2^{R_s}+\rho-1)\sigma_{b}^2}{(1-2^{R_s})(1-\rho)P_a}\right)\!\times\! \int_{0}^{\frac{(1-2^{R_s}(1-\rho)) (\sigma_{w}^2+\sigma_{e}^2)}{(2^{R_s}-1)(1-\rho)P_a}}\!\! \int_{0}^{\frac{(1-2^{R_s}(1-\rho)) (\sigma_{w}^2+\sigma_{e}^2)}{(2^{R_s}-1)(1-\rho)P_a}-z} \nonumber \\
& \ \ \ \!\!  \times \exp \left(-y-\frac{(2^{R_s}-1)\sigma_{b}^2(\sigma_{w}^2+\sigma_{e}^2)-\frac{(2^{R_s}+\rho-1)(1-(1-\rho)2^{R_s})\sigma_{b}^2(\sigma_{w}^2+\sigma_{e}^2)}{(1-2^{R_s})(1-\rho)}}{(1-2^{R_s})(1-\rho)P_a^2 (y+z)+ (1-(1-\rho)2^{R_s}) P_a(\sigma_{w}^2+\sigma_{e}^2)}-z\right) \text{\,d}y \text{\,d}z,
\end{align}
{\noindent} 
\rule[-10pt]{18.07cm}{0.05em}
\end{figure*}

When Alice does not transmit messages, we consider only the covertness of the transmission by analyzing the probability of false alarm. In this case, Alice still sends AN to confuse Willie. Thus, based on (\ref{eqn:yW-half-ind-an}), the signal vector $\mathbf y_{w}$ contains both the signals (i.e., AN and background noise) shared by Eve, AN and background noise.
In this case, the average power of the received symbols at Willie is $\bar{P}_w =(1-\rho)P_a|h_{aw}|^2+(1-\rho)P_a|h_{ae}|^2+ \sigma_e^2+ \sigma_w^2$.
Thus, the probability of false alarm $p_{FA}$ is given by
\begin{align}\label{pfa_fri_an}
p_{FA} & \!\!=\!\! \mathbb{P} \left((\!1\!\!-\!\!\rho\!)P_a|h_{aw}|^2\!+\! (\!1\!\!-\!\!\rho\!)P_a|h_{ae}|^2\!+\! \sigma_e^2\!+\! \sigma_w^2 \!\ge\!\theta \right) \\
& \!\!=\!\!
\begin{cases}
\left(\!1\!+\!\frac{\theta-\sigma_{e}^2- \sigma_{w}^2}{(1-\rho)P_a}\!\right)\!\exp\!\left(\!-\frac{\theta-\sigma_{e}^2- \sigma_{w}^2}{(1-\rho)P_a}\!\right), & \theta > \sigma_{e}^2+ \sigma_{w}^2, \\
1, & \theta \le \sigma_{e}^2+ \sigma_{w}^2.
\end{cases}\nonumber
\end{align}

Combining the $p_{FA}$ in (\ref{pfa_fri_an}) and the $p_{MD}$ in (\ref{pmd_fri_pc}), the COP can be given by
\begin{equation} \label{pco_fri_an}
p_{co}^{\mathrm{FA}}(\rho,\!\theta) \!\!=\!\!
\begin{cases}
\!\left(1\!+\!\frac{\theta-\sigma_{e}^2- \sigma_{w}^2}{P_a}\right)\!\exp\!\left(\!- \frac{\theta-\sigma_{e}^2- \sigma_{w}^2}{P_a}\right)&\\
-\!\left(\!1\!\!+\!\!\frac{\theta-\sigma_{e}^2\!- \!\sigma_{w}^2}{(1\!-\!\rho)P_a}\right)\!\exp\!\left(\!-\frac{\theta-\sigma_{e}^2\!- \!\sigma_{w}^2}{(1\!-\!\rho)P_a}\right),&\!\! \theta\! >\! \sigma_{e}^2\!+\! \sigma_{w}^2, \\
0, &\!\! \theta \!\le\! \sigma_{e}^2\!+\! \sigma_{w}^2.
\end{cases}
\end{equation}

We can see from (\ref{pco_fri_an}) that the optimal detection threshold $\theta^*_{\mathrm{FA}}$ can be obtained by solving $\frac{\partial p_{co}^{\mathrm{FA}}}{\partial \theta}=0$, which is
\begin{equation}\label{theta_fri_an}
\theta^*_{\mathrm{FA}}=\sigma_{e}^2+ \sigma_{w}^2+\frac{2(\rho -1)P_a}{\rho} \ln(1-\rho).
\end{equation}

Given the $\theta^*_{\mathrm{FA}}$ in (\ref{theta_fri_an}), we solve the optimization problem in (\ref{AN-based}) to obtain the CSR, which is given in the following theorem.
\begin{theorem}
Under the scenario where Willie and Eve are in the friend relationship and Alice adopts the AN-based secure transmission scheme, the CSR of the system is
\begin{equation}
R_{cs}^{\mathrm{FA}} \!\!=\!\!R^*_{s,\mathrm{FA}}\!(\!\rho^*_{\mathrm{FA}}\!)\exp\!\! \left(\!- \frac{(2^{R^*_{s,\mathrm{FA}}(\rho^*_{\mathrm{FA}})}-1) \sigma_{b}^2}{\rho^*_{\mathrm{FA}}\! P_a\!-\! (\!2^{R^*_{s,\mathrm{FA}}\!(\!\rho^*_{\mathrm{FA}}\!)}\!\!-\!\!1\!)(\!1\!\!-\!\!\rho^*_{\mathrm{FA}}\!)\!P_a}\!\right)\!.
\end{equation}
Here, the optimal power allocation parameter $\rho^*_{\mathrm{FA}}$ solves $p_{co}^{\mathrm{FA}} (\rho, \theta^*_{\mathrm{FA}})= \epsilon_c$ with $\theta^*_{\mathrm{FA}}$ given by (\ref{theta_fri_an}).
The optimal secrecy rate $R^*_{s,\mathrm{FA}}$ is given in (\ref{rs_opt_IA}), where $R_{s,\mathrm{FA}}^0$ can be obtained by solving $\frac{\partial R_{cs}}{\partial R_s}=0$, $R^{\text{SOP}}_{s,\mathrm{FA}}$ is the solution of $p_{so}^{\mathrm{FA}}(R_s)=\epsilon_s$ and $R^{\text{TP}}_{s,\mathrm{FA}}$ is given in (\ref{rs_top_IA}).
\end{theorem}

\begin{proof}
The proof is similar to that of Theorem \ref{theorem:csr-ind-an} and thus omitted here.
\end{proof}

\section{Numerical Results} \label{5}
\begin{figure}[!t]
\centering
\subfigure[Independence relationship scenario.]{
	\includegraphics[width=0.4\textwidth]{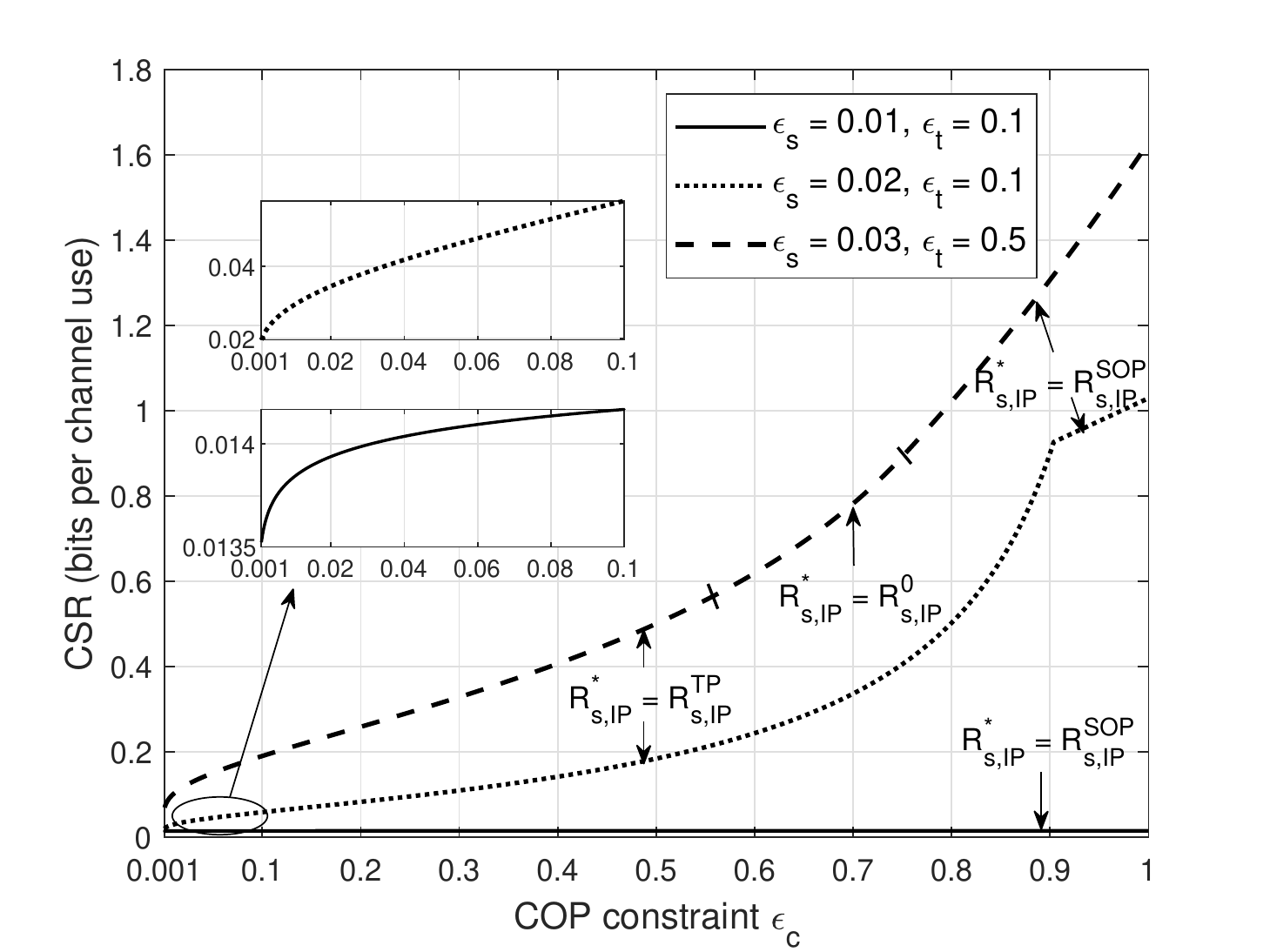}
	\label{IP_Rcs_Ec}}
\subfigure[Friend relationship scenario.]{
	\includegraphics[width=0.4\textwidth]{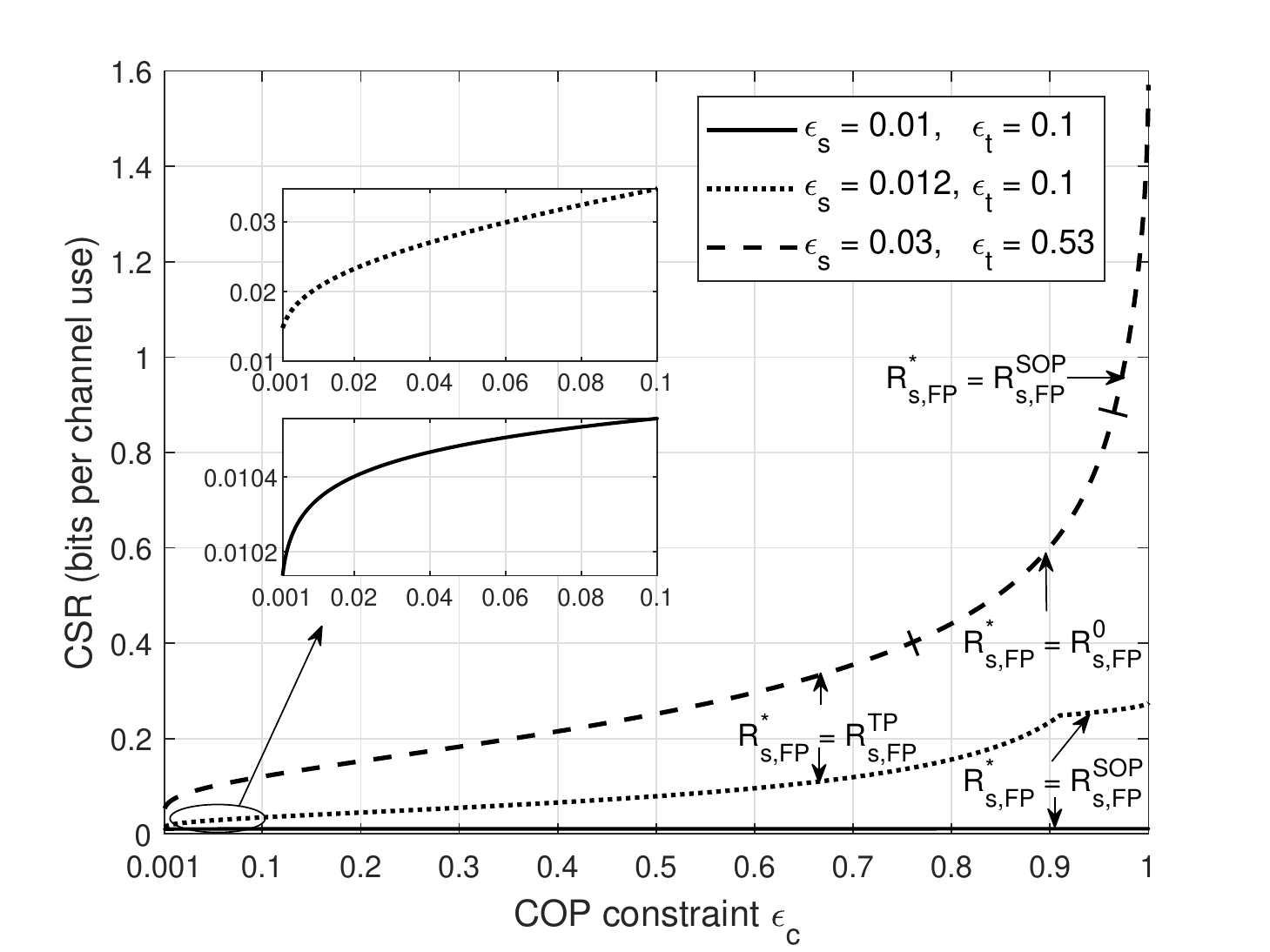}
	\label{FP_Rcs_Ec}}
\caption{CSR  $R_{cs}$ vs. COP constraint $\epsilon_c$ (PC-based transmission scheme).}
\label{Rcs_Ec_PC}
\end{figure}

\begin{figure}[!t]
\centering
\subfigure[Independence relationship scenario.]{
	\includegraphics[width=0.4\textwidth]{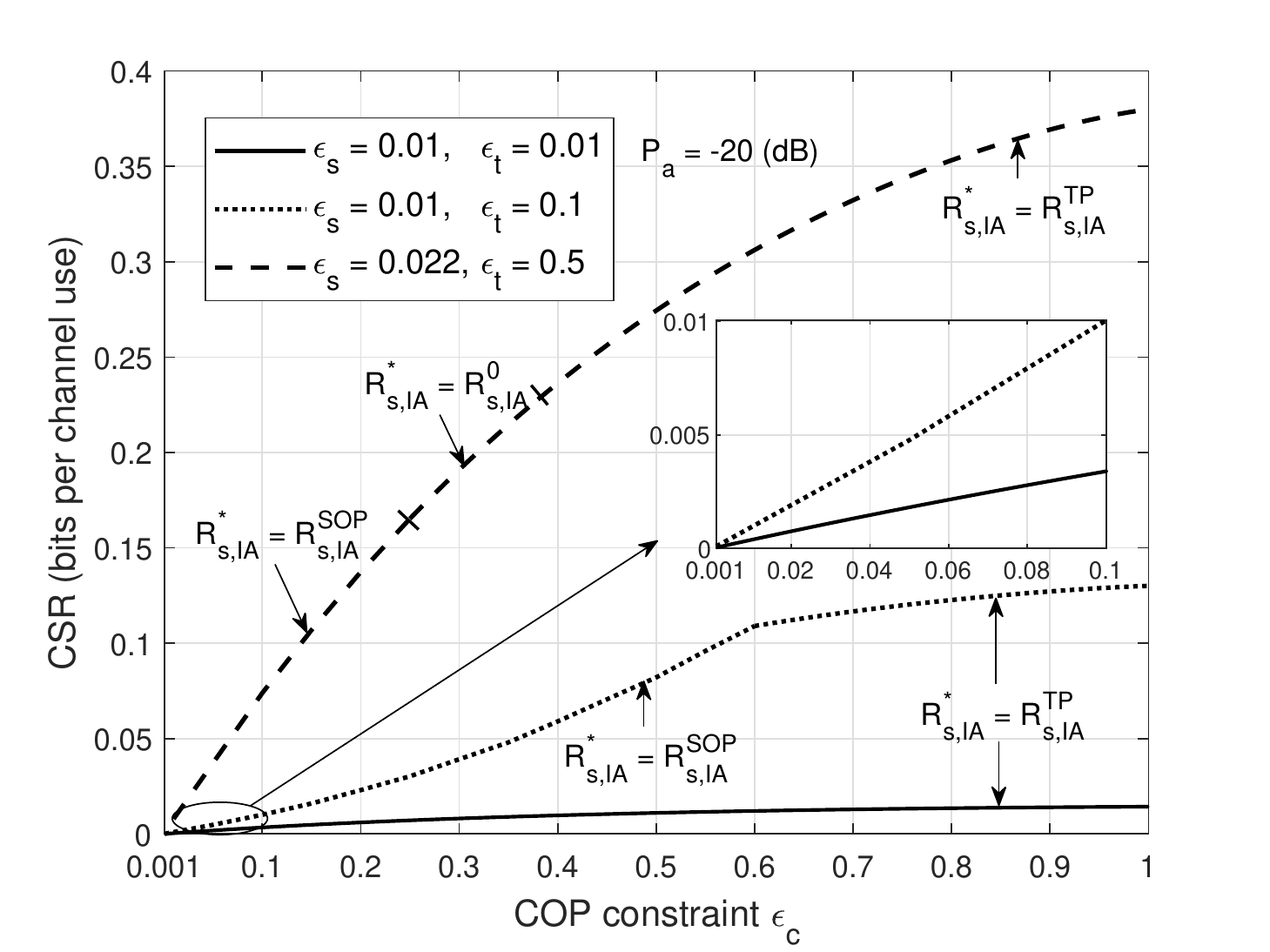}
	\label{IA_Rcs_Ec}}
\subfigure[Friend relationship scenario.]{
	\includegraphics[width=0.4\textwidth]{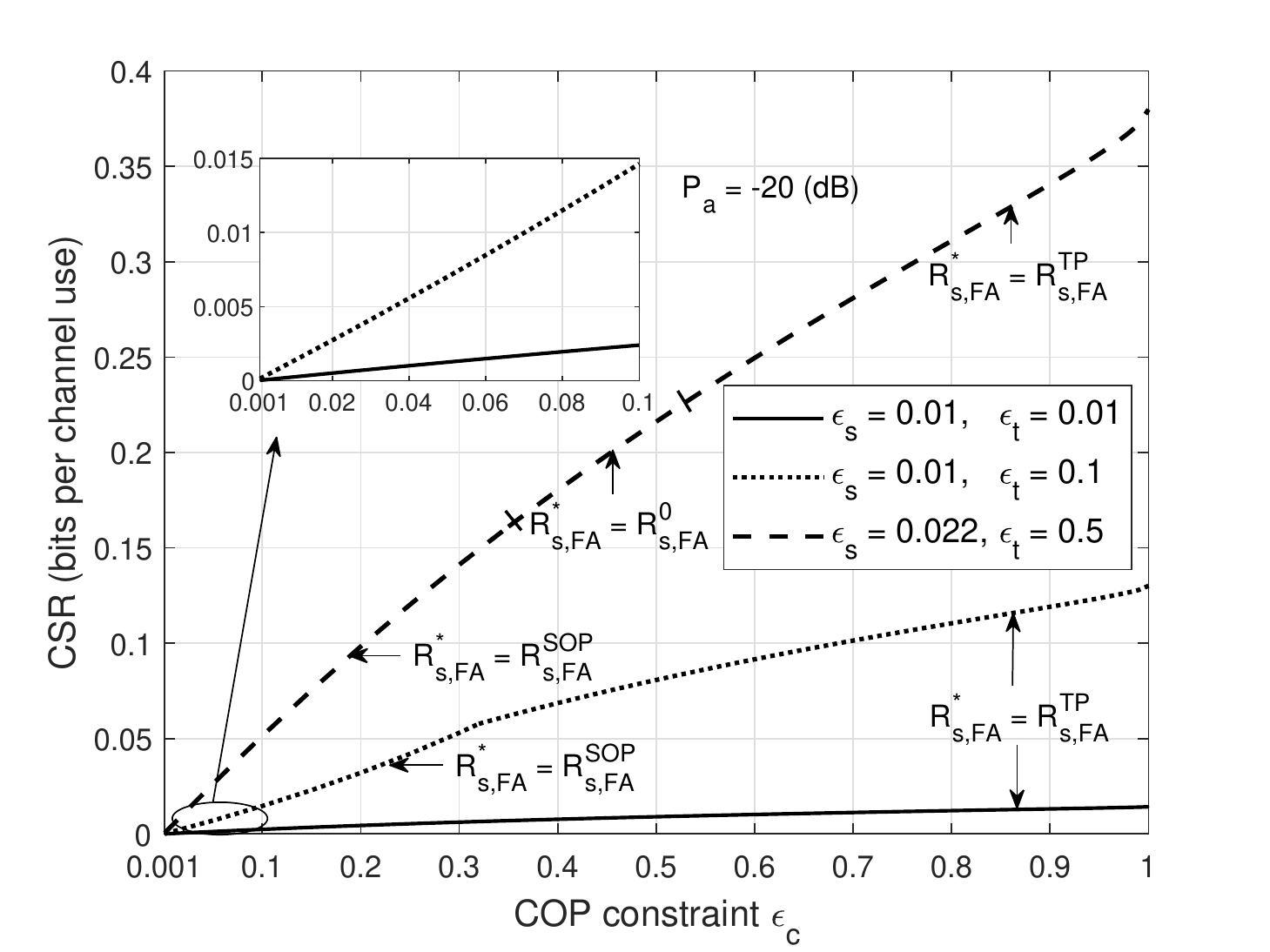}
	\label{FA_Rcs_Ec}}
\caption{CSR  $R_{cs}$ vs. COP constraint $\epsilon_c$ (AN-based transmission scheme).}
\label{Rcs_Ec_AN}
\end{figure}

In this section, we provide extensive numerical results to illustrate the CSR performances of the four representative scenarios under the new secure communication paradigm.
We also show the impacts of various system parameters (e.g., COP constraint $\epsilon_c$, SOP constraint $\epsilon_s$, TP constraint $\epsilon_t$ and transmit power $P_a$) on the CSR performance.
Unless otherwise stated, we set the parameter $\upsilon$ to $\upsilon=0.01$ and the noise powers at Bob, Willie and Eve to $\sigma_b^2=-20$ dB and $\sigma_w^2=\sigma_e^2=0$ dB.

To explore the impact of the COP constraint $\epsilon_c$ on the CSR performance, we show in Fig. \ref{Rcs_Ec_PC}  $R_{cs}$ vs. $\epsilon_c$ in the independence relationship case under the PC-based and AN-based transmission schemes, respectively. The results for the friend relationship case under both transmission schemes are presented in Fig. \ref{Rcs_Ec_AN}. We set the transmit power of Alice to $P_a=-20$ dB in Fig. \ref{Rcs_Ec_AN}.
In each subfigure of Fig. \ref{Rcs_Ec_PC} and Fig. \ref{Rcs_Ec_AN}, we also plot the CSR curves under different settings of SOP constraint $\epsilon_s$ and TP constraint $\epsilon_t$.
We can see from Fig. \ref{Rcs_Ec_PC} and Fig. \ref{Rcs_Ec_AN} that the CSRs achieved under different SOP and TP constraints always increase as $\epsilon_c$ increases.
This is because a looser COP constraint results in a larger optimal transmit power in the PC-based scheme (resp. a larger optimal power allocation parameter in the AN-based scheme) and thus a larger CSR.

We can also observe from Fig. \ref{Rcs_Ec_PC} and Fig. \ref{Rcs_Ec_AN} that the shape of the CSR curve varies as the values of the SOP constraint $\epsilon_s$ and TP constraint $\epsilon_t$ change. For example, the CSR curve under the setting of $\epsilon_s=0.03$ and $\epsilon_t=0.5$ (dashed line) in Fig. \ref{Rcs_Ec_PC} exhibits an exponential growth and that under the setting of $\epsilon_s=0.02$ and $\epsilon_t=0.1$ (dotted line) grows in a piecewise fashion.
This is because different values of $\epsilon_s$, $\epsilon_t$ and the COP constraint $\epsilon_c$ result in different $R_{s,\mathrm{IP}}^\text{SOP}$, $R_{s,\mathrm{IP}}^\text{TP}$ and $R_{s,\mathrm{IP}}^0$ in (\ref{IP_Rs_sop}-\ref{IP_Rs_Rcs}) (resp. $R_{s,\mathrm{FP}}^\text{SOP}$, $R_{s,\mathrm{FP}}^\text{TP}$, $R_{s,\mathrm{FP}}^0$ in (\ref{FP_Rs_sop},\ref{IP_Rs_top},\ref{IP_Rs_Rcs}), $R_{s,\mathrm{IA}}^\text{SOP}$, $R_{s,\mathrm{IA}}^\text{TP}$, $R_{s,\mathrm{IA}}^0$ in (\ref{rs_opt_IA}) and $R_{s,\mathrm{FA}}^\text{SOP}$, $R_{s,\mathrm{FA}}^\text{TP}$, $R_{s,\mathrm{FA}}^0$ in (\ref{rs_opt_IA})), which further lead to different optimal target secrecy rates (as labeled in Fig. \ref{Rcs_Ec_PC} and Fig. \ref{Rcs_Ec_AN}) and thus different CSR curves.

\begin{figure}[!t]
\centering
\subfigure[Independence relationship scenario.]{
	\includegraphics[width=0.4\textwidth]{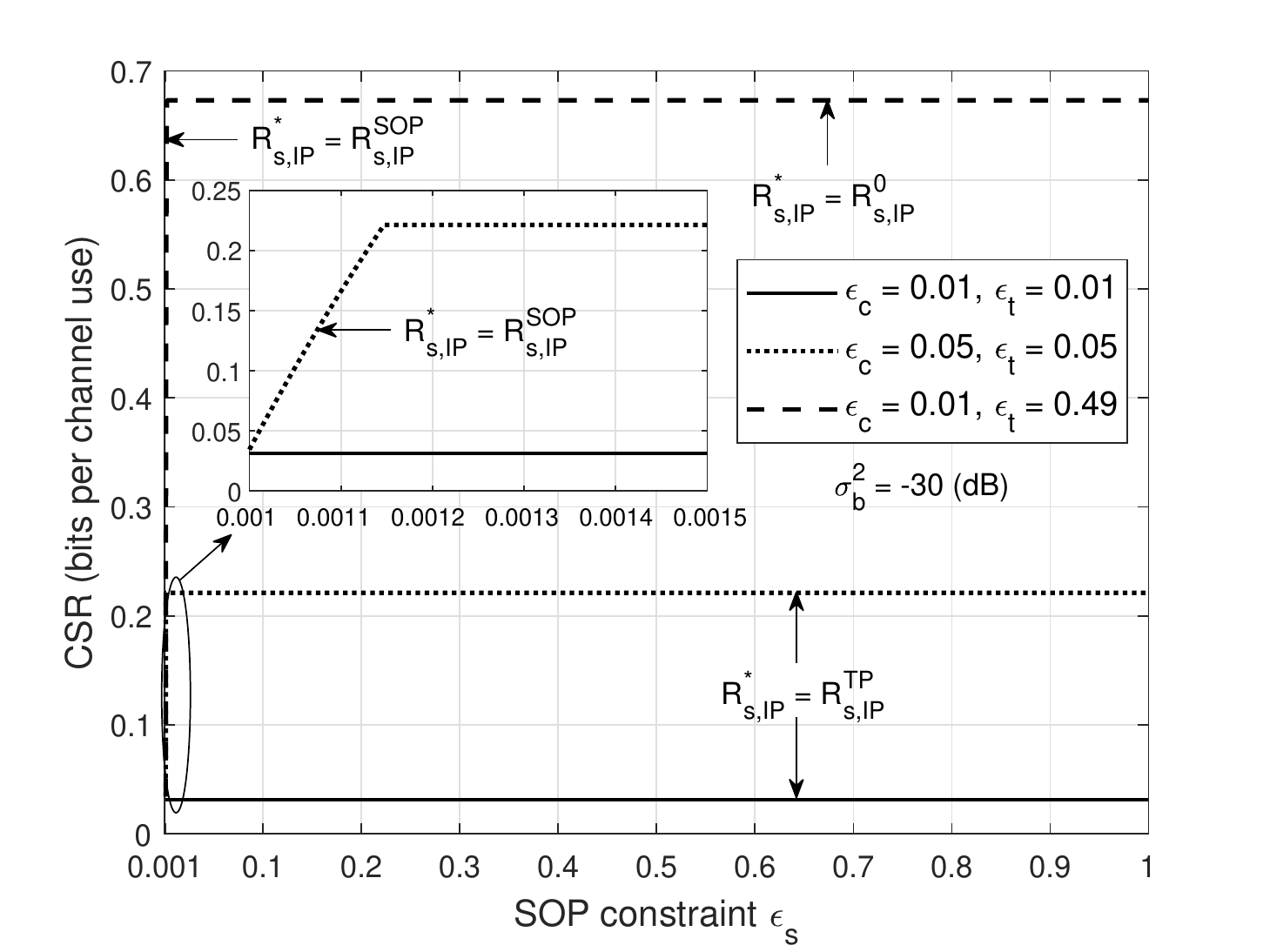}
	\label{IP_Rcs_Es}}
\subfigure[Friend relationship scenario.]{
	\includegraphics[width=0.4\textwidth]{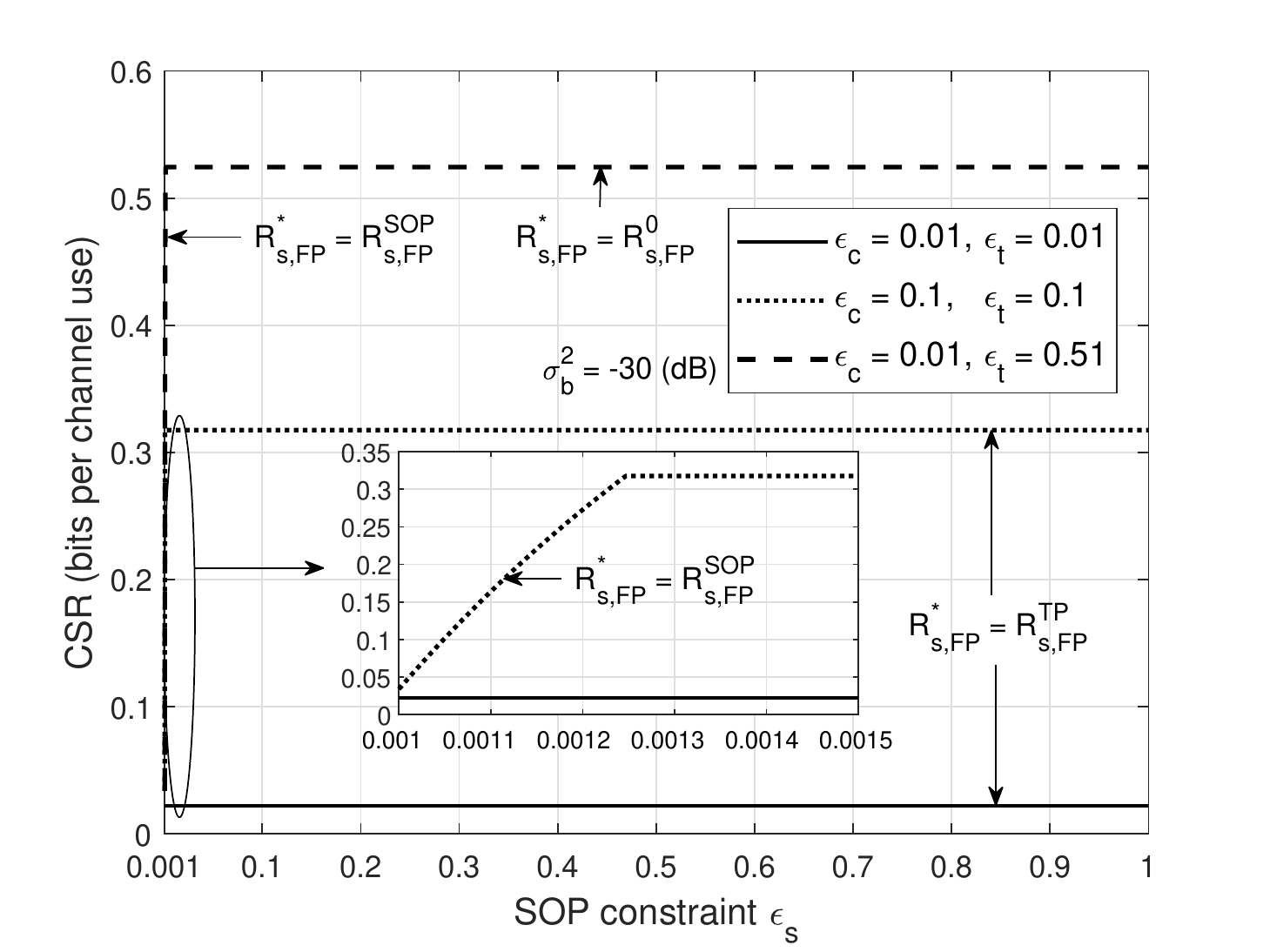}
	\label{FP_Rcs_Es}}
\caption{CSR $R_{cs}$ vs. SOP constraint $\epsilon_s$ (PC-based transmission scheme).}
\label{epsilon_s_PC}
\end{figure}

\begin{figure}[!t]
\centering
\subfigure[Independence relationship scenario.]{
	\includegraphics[width=0.4\textwidth]{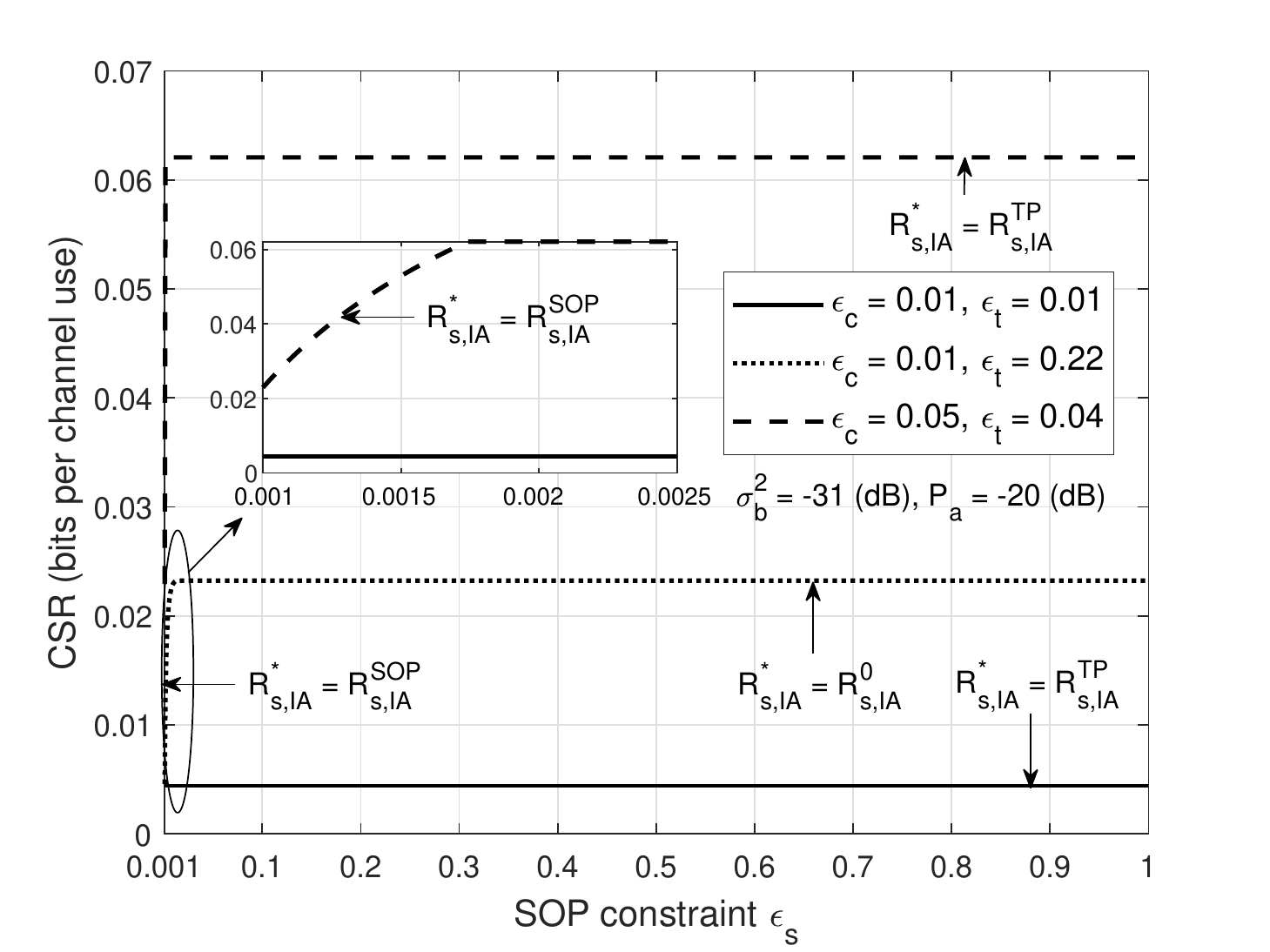}
	\label{IA_Rcs_Es}}
\subfigure[Friend relationship scenario.]{
	\includegraphics[width=0.4\textwidth]{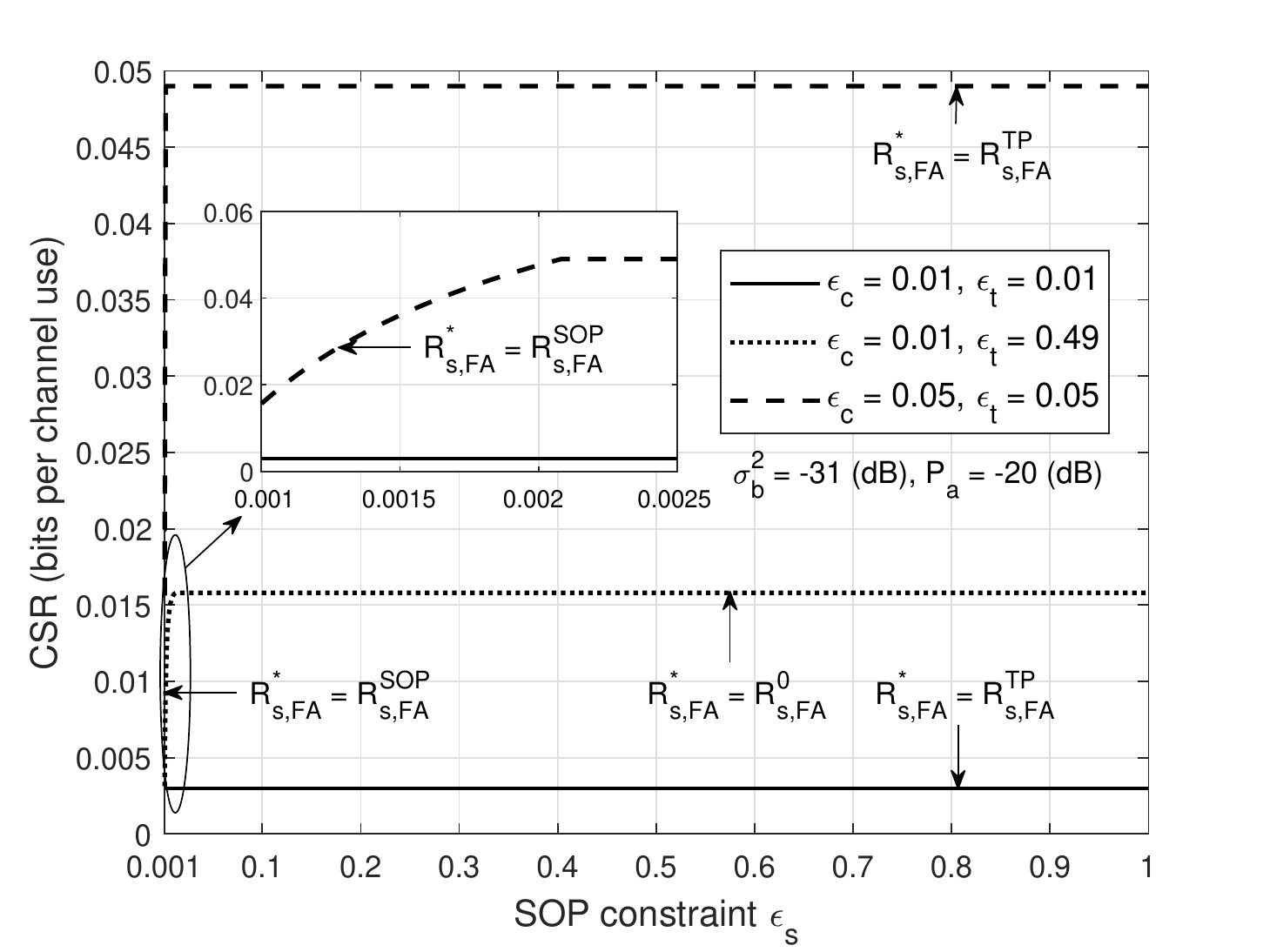}
	\label{FA_Rcs_Es}}
\caption{CSR $R_{cs}$ vs. SOP constraint $\epsilon_s$ (AN-based transmission scheme).}
\label{epsilon_s_AN}
\end{figure}

Next, we investigate the impact of the SOP constraint $\epsilon_s$ on the CSR performance, for which we show $R_{cs}$ vs. $\epsilon_s$ in the independence and friend relationship cases under the PC-based transmission scheme in Fig. \ref{epsilon_s_PC} and those under the AN-based transmission scheme in Fig. \ref{epsilon_s_AN}. We set the noise power at Bob to $\sigma_b^2=-30$ dB in Fig. \ref{epsilon_s_PC} and that to $\sigma_b^2=-31$ dB in Fig. \ref{epsilon_s_AN}. We set the transmit power of Alice to $P_a=-20$ dB in Fig. \ref{epsilon_s_AN}. For both figures, we consider three different settings of COP constraint $\epsilon_c$ and TP constraint $\epsilon_t$, respectively.
We can see from Fig. \ref{epsilon_s_PC} and Fig. \ref{epsilon_s_AN} that, when both $\epsilon_c$ and $\epsilon_t$ are relatively small (e.g., $\epsilon_c=0.01$ and $\epsilon_t=0.01$ in Fig. \ref{IP_Rcs_Es}), the CSR stays unchanged as the SOP constraint $\epsilon_s$ increases, which implies that the SOP constraint $\epsilon_s$ has no impacts on the CSR performance. This is because, in this situation, the CSR is achieved at only the optimal target secrecy rate $R_{s,\mathrm{IP}}^*=R_{s,\mathrm{IP}}^\text{TP}$ (as labeled in Fig. \ref{IP_Rcs_Es}), which is independent of $\epsilon_s$ as can be seen from \eqref{IP_Rs_top}.
On the other hand, when either $\epsilon_c$ or $\epsilon_t$ is large, the CSR first increases sharply and then remains constant as the SOP constraint $\epsilon_s$ increases. This is because the optimal target secrecy rate is $R_{s,\mathrm{IP}}^*=R_{s,\mathrm{IP}}^\text{SoP}$ for small $\epsilon_s$, which increases as $\epsilon_s$ increases, and then changes to $R_{s,\mathrm{IP}}^*=R_{s,\mathrm{IP}}^\text{0}$ or $R_{s,\mathrm{IP}}^*=R_{s,\mathrm{IP}}^\text{TP}$ for large $\epsilon_s$, which is independent of $\epsilon_s$. Such phenomenon indicates that, when either $\epsilon_c$ or $\epsilon_t$ is large, the CSR is sensitive to the change of the SOP constraint $\epsilon_s$ in an extremely small region, e.g., from $0$ to about $0.00115$ in Fig.~\ref{IP_Rcs_Es}. Similar phenomena can be observed from Fig.~\ref{FP_Rcs_Es}, Fig.~\ref{IA_Rcs_Es} and Fig.~\ref{FA_Rcs_Es}.

\begin{figure}[!t]
\centering
\subfigure[Independence relationship scenario.]{
	\includegraphics[width=0.4\textwidth]{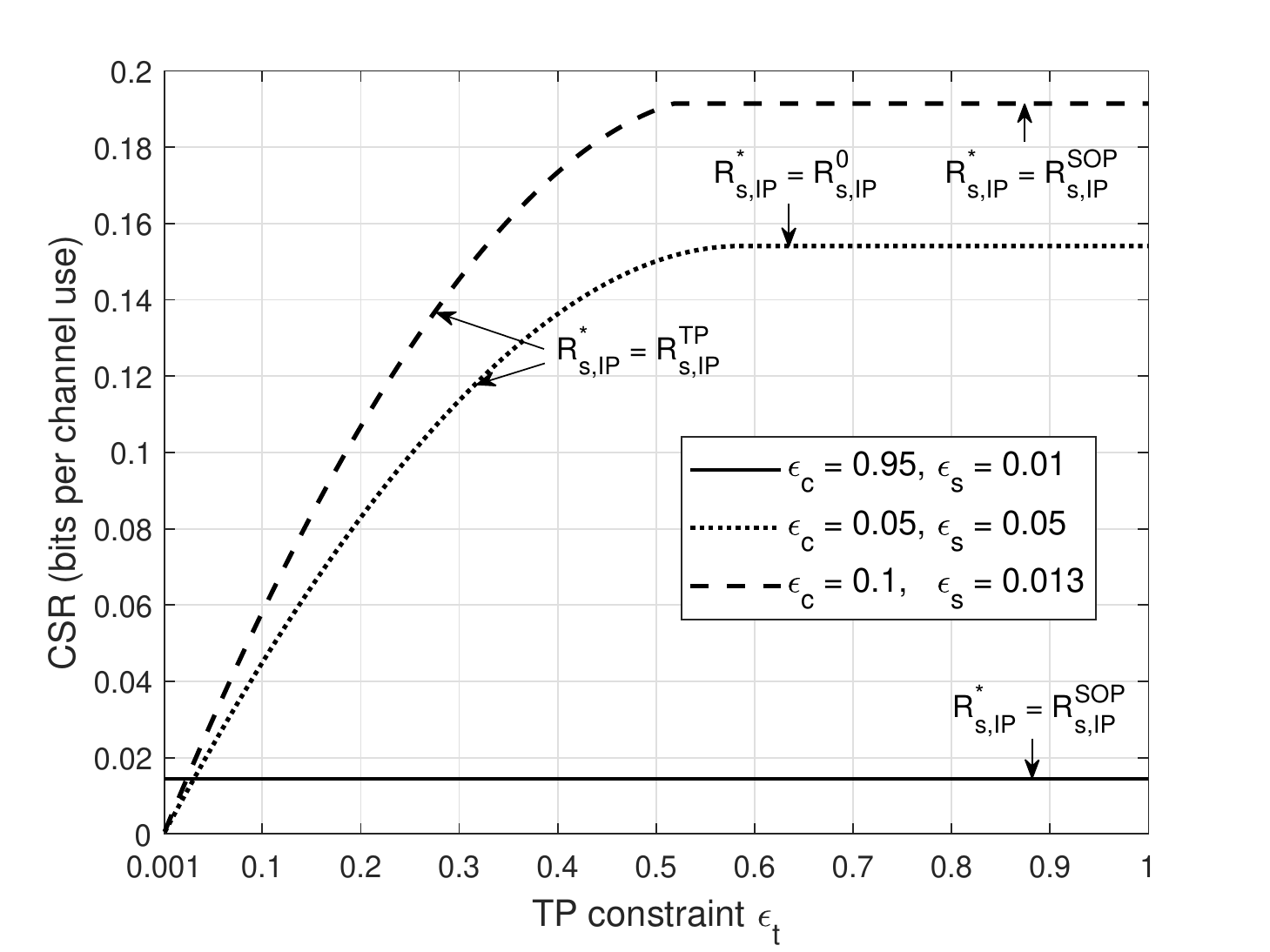}
	\label{IP_Rcs_Et}}
\subfigure[Friend relationship scenario.]{
	\includegraphics[width=0.4\textwidth]{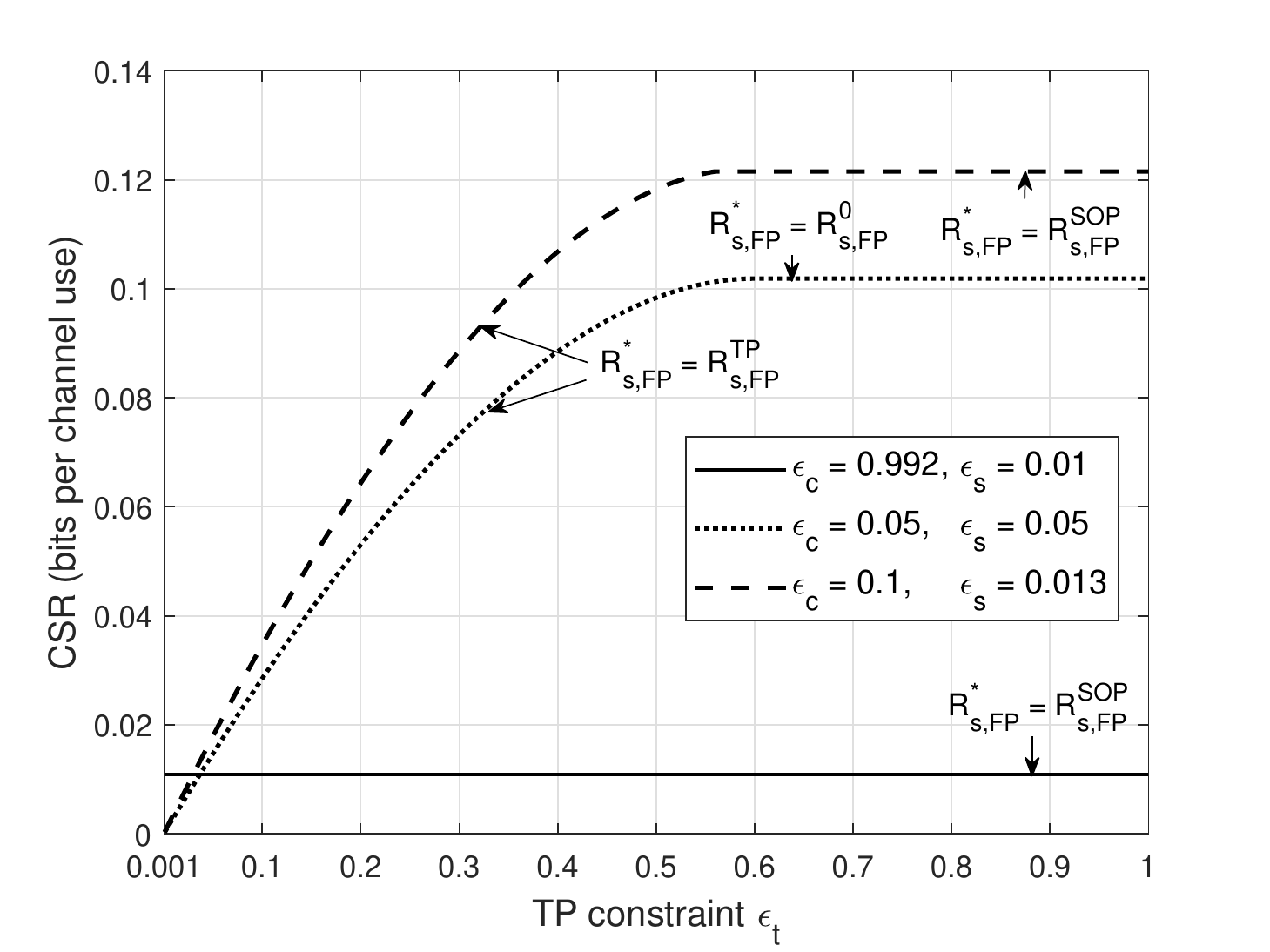}
	\label{FP_Rcs_Et}}
\caption{CSR $R_{cs}$ vs. TP constraint $\epsilon_t$ (PC-based transmission scheme).}
\label{epsilon_t_PC}
\end{figure}

\begin{figure}[!t]
\centering
\subfigure[Independence relationship scenario.]{
	\includegraphics[width=0.4\textwidth]{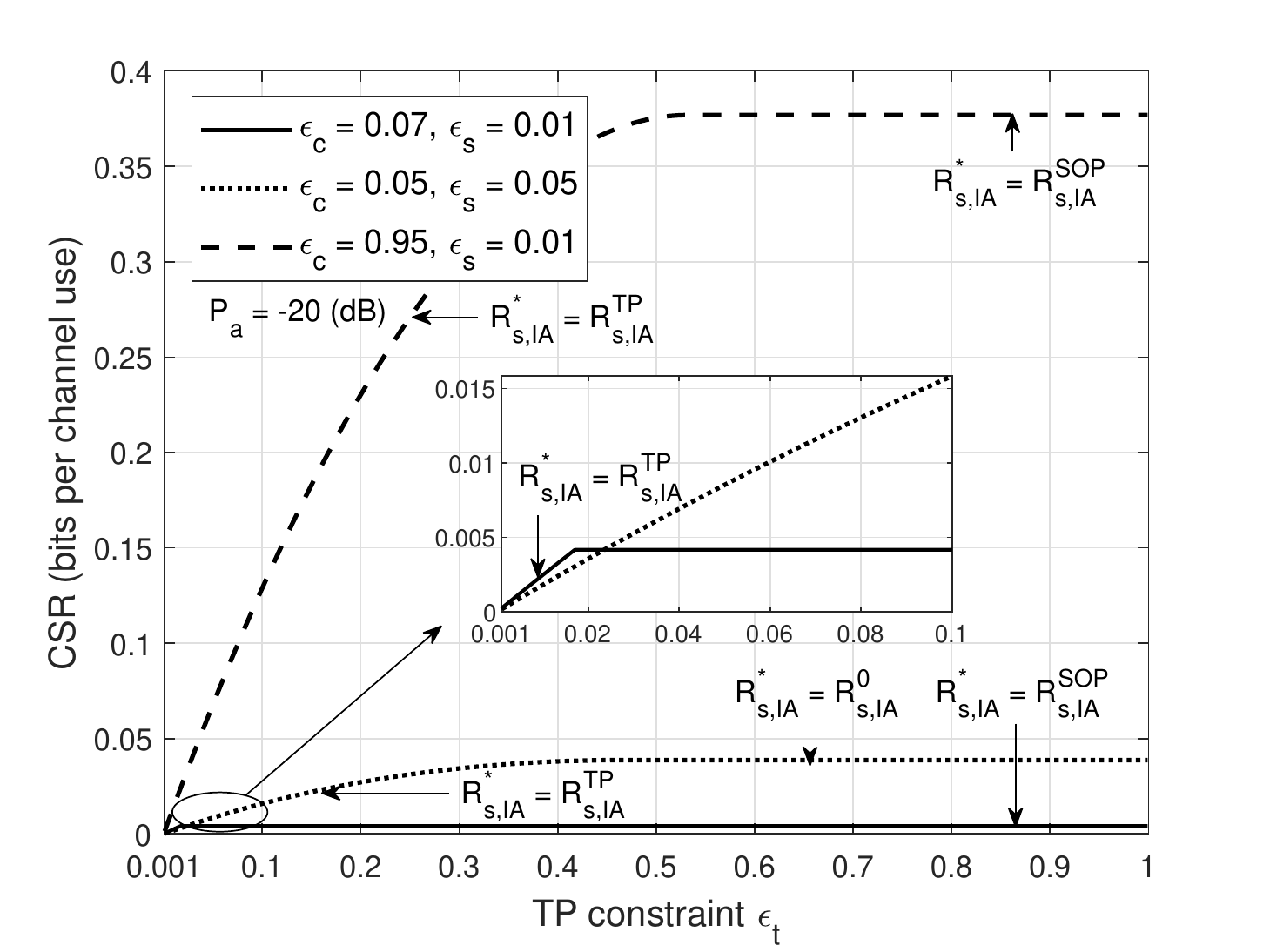}
	\label{IA_Rcs_Et}}
\subfigure[Friend relationship scenario.]{
	\includegraphics[width=0.4\textwidth]{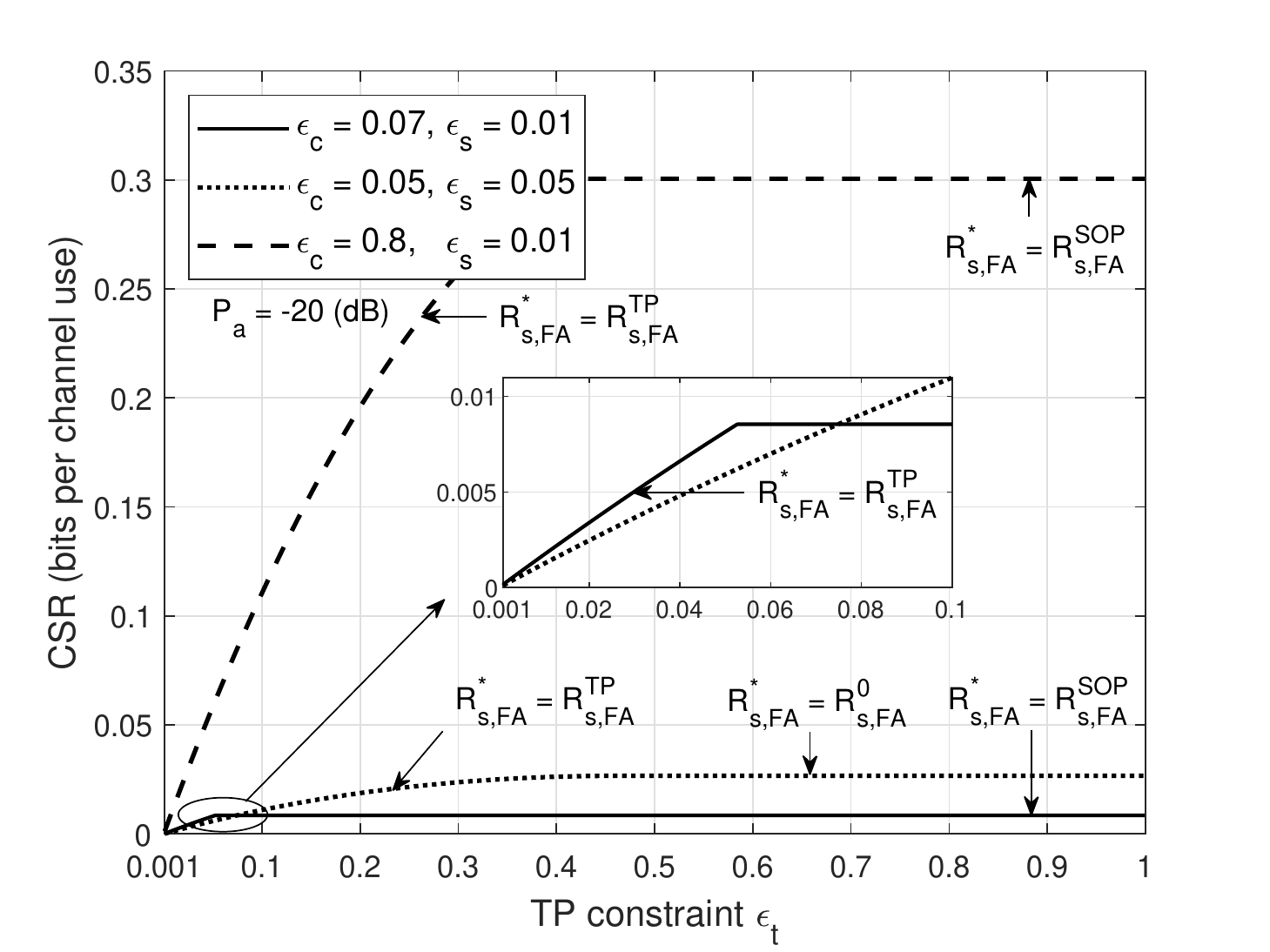}
	\label{FA_Rcs_Et}}
\caption{CSR $R_{cs}$ vs. TP constraint $\epsilon_t$ (AN-based transmission scheme).}
\label{epsilon_t_AN}
\end{figure}

We now show the impact of the TP constraint $\epsilon_t$ on the CSR performance in Fig. \ref{epsilon_t_PC} and Fig. \ref{epsilon_t_AN}, where we plot $R_{cs}$ vs. $\epsilon_t$ for the two relationship cases under the PC-based and AN-based transmission schemes, respectively. Three different settings of COP constraint $\epsilon_c$ and SOP constraint $\epsilon_s$ are adopted for each subfigure in Fig. \ref{epsilon_t_PC} and Fig. \ref{epsilon_t_AN}. We set the transmit power of Alice to $P_a=-20$ dB in Fig. \ref{epsilon_t_AN}. We can see from Fig. \ref{IP_Rcs_Et} that, if the COP constraint $\epsilon_c$ is much larger than the SOP constraint $\epsilon_s$, the CSR stays constant as the constraint $\epsilon_t$ increases, i.e., the CSR is independent of $\epsilon_t$. Otherwise, the CSR first increases and then stays constant as $\epsilon_t$ increases. This is because, for the former case, the CSR is achieved at only the optimal target secrecy rate $R_{s,\mathrm{IP}}^*=R_{s,\mathrm{IP}}^\text{SOP}$ (as labeled in Fig. \ref{IP_Rcs_Et}), which is independent of $\epsilon_t$ as can be seen from \eqref{IP_Rs_sop}. For the latter case, the optimal target secrecy rate is $R_{s,\mathrm{IP}}^*=R_{s,\mathrm{IP}}^\text{TP}$ for small $\epsilon_t$, which increases as $\epsilon_t$ increases, and then changes to $R_{s,\mathrm{IP}}^*=R_{s,\mathrm{IP}}^\text{0}$ or $R_{s,\mathrm{IP}}^*=R_{s,\mathrm{IP}}^\text{SOP}$ for large $\epsilon_t$, which is independent of $\epsilon_t$. We can observe similar phenomena from Fig.~\ref{FP_Rcs_Et}, Fig.~\ref{IA_Rcs_Et} and Fig.~\ref{FA_Rcs_Et}.

\begin{figure}[!t]
\centering
\subfigure[PC-based transmission scheme.]{
	\includegraphics[width=0.4\textwidth]{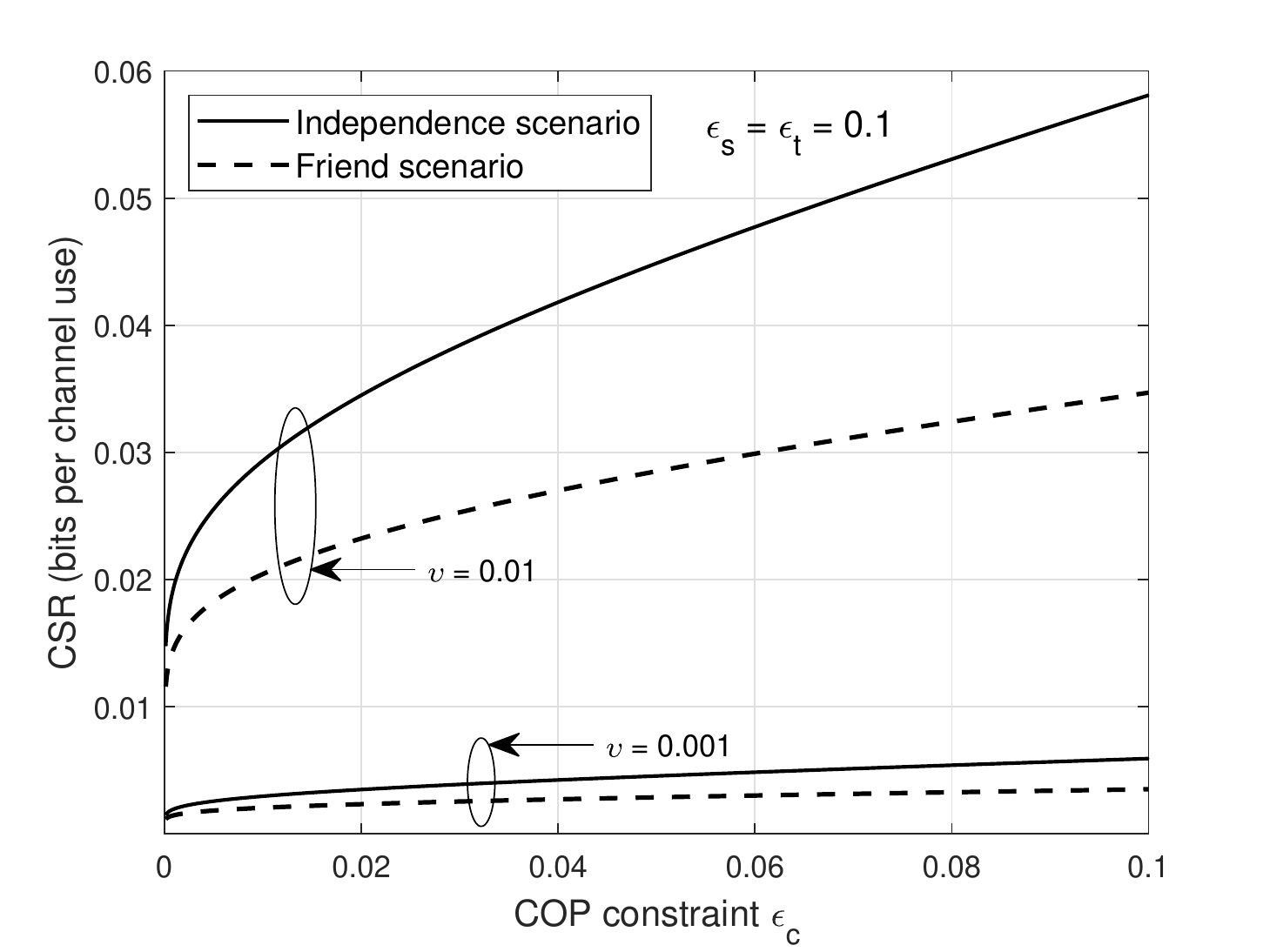}
	\label{PC_Rcs_Ec}}
\subfigure[AN-based transmission scheme.]{
	\includegraphics[width=0.4\textwidth]{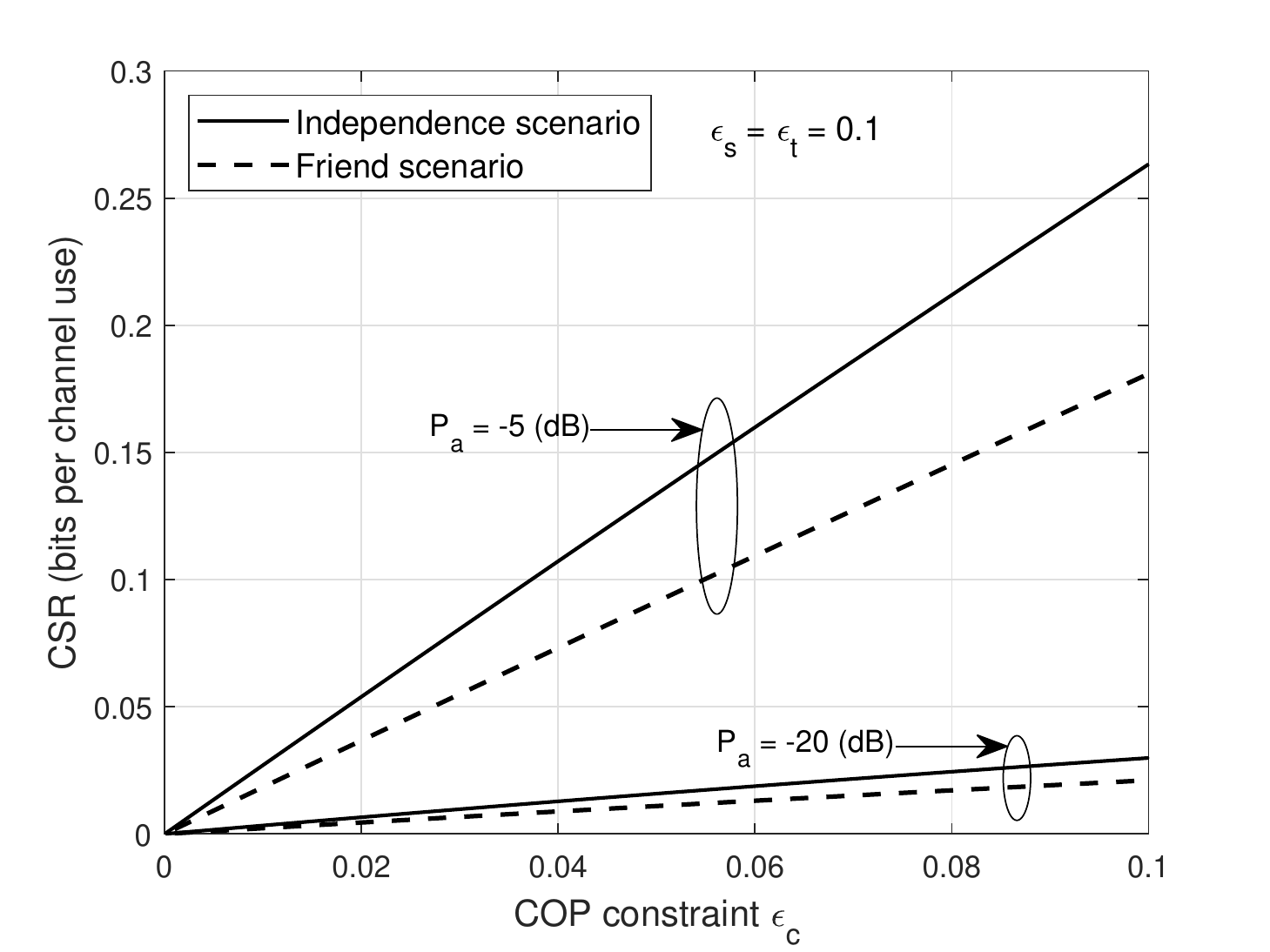}
	\label{AN_Rcs_Ec}}
\caption{Comparisons of the CSR performances in two relationship cases.}
\label{scenario}
\end{figure}

We proceed to compare the CSR performance achieved in the independence relationship scenario and that achieved in the friend relationship scenario, for which we show $R_{cs}$ vs. $\epsilon_c$ for both relationship scenarios under the PC-based transmission scheme in Fig. \ref{PC_Rcs_Ec} and those under the AN-based transmission scheme in Fig. \ref{AN_Rcs_Ec}, respectively. We set the SOP constraint and TP constraint to $\epsilon_s=\epsilon_t=0.1$ in both figures. In addition, we set the parameter $\upsilon$ to $\upsilon=0.01$ and $0.001$ in Fig. \ref{PC_Rcs_Ec} and the transmit power of Alice $P_a$ to $P_a=-5$ dB and $-20$ dB in Fig. \ref{AN_Rcs_Ec}. We can observe from both subfigures that the CSRs in the independence relationship case are always larger than those in the friend relationship case under all the parameter settings and both transmission schemes.
This is intuitive since Willie and Eve can improve their attacking abilities by sharing their signals.
The above observations indicate that being friends is the better choice than being independent for the eavesdropper group and detector group.

\begin{figure}[!t]
\centering
\subfigure[Independence relationship scenario.]{
	\includegraphics[width=0.4\textwidth]{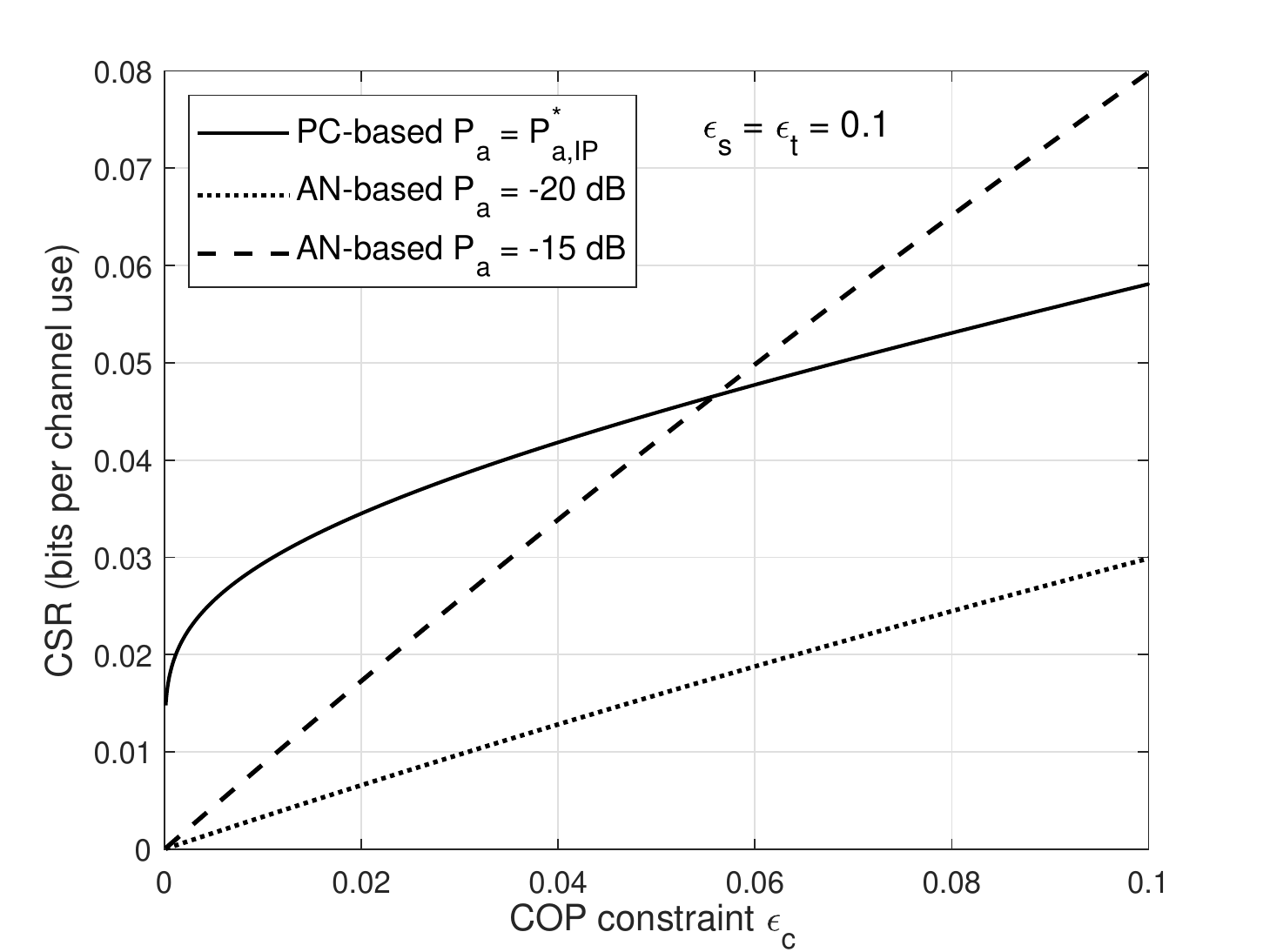}
	\label{Ind_Rcs_Ec}}
\subfigure[Friend relationship scenario.]{
	\includegraphics[width=0.4\textwidth]{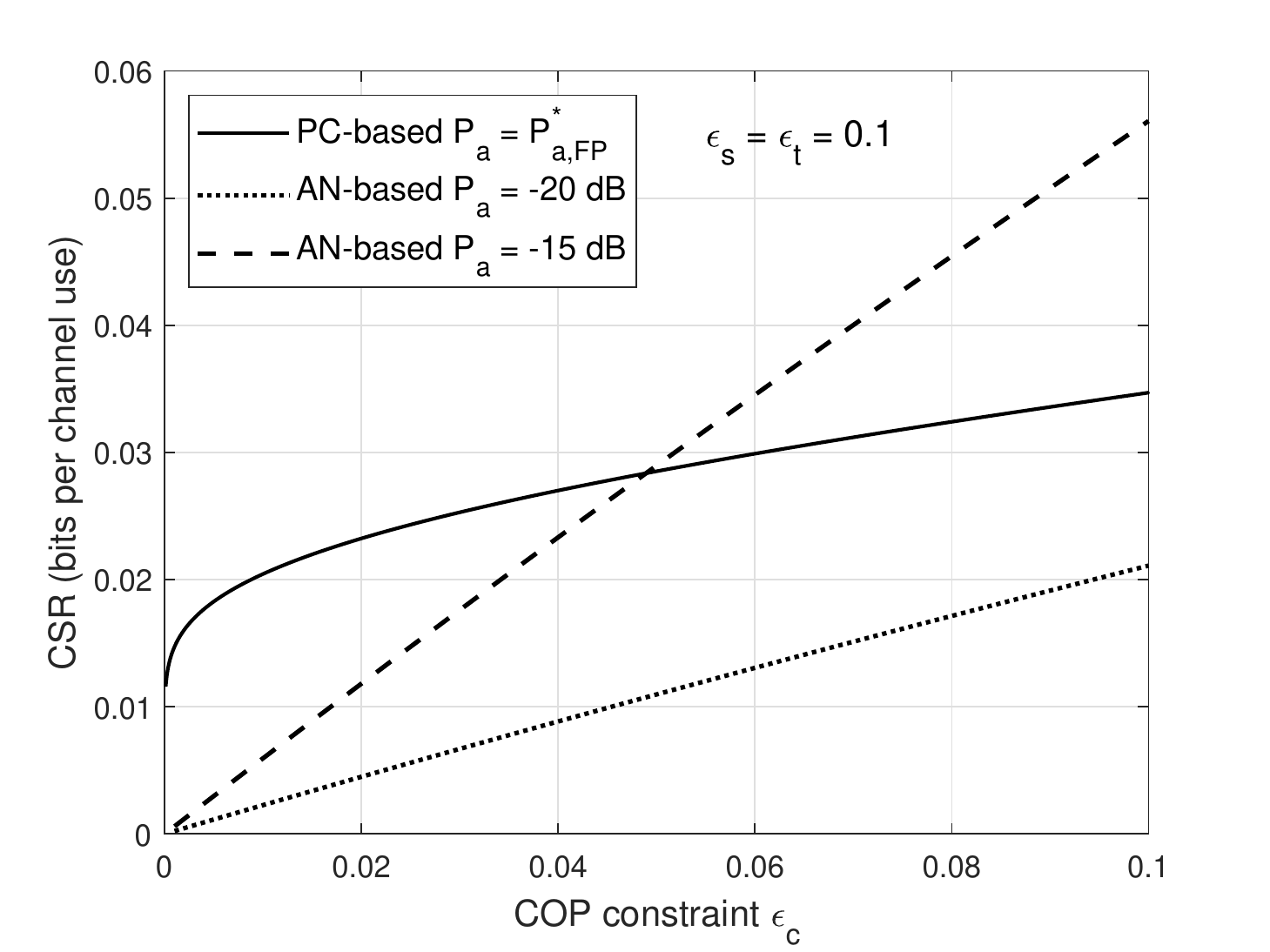}
	\label{Fri_Rcs_Ec}}
\caption{Comparisons of the CSR performances in the PC-based and AN-based transmission schemes ($R_{cs}$ vs. $\epsilon_c$).}
\label{transmission_Ec}
\end{figure}

\begin{figure}[!t]
\centering
\subfigure[Independence relationship scenario.]{
	\includegraphics[width=0.4\textwidth]{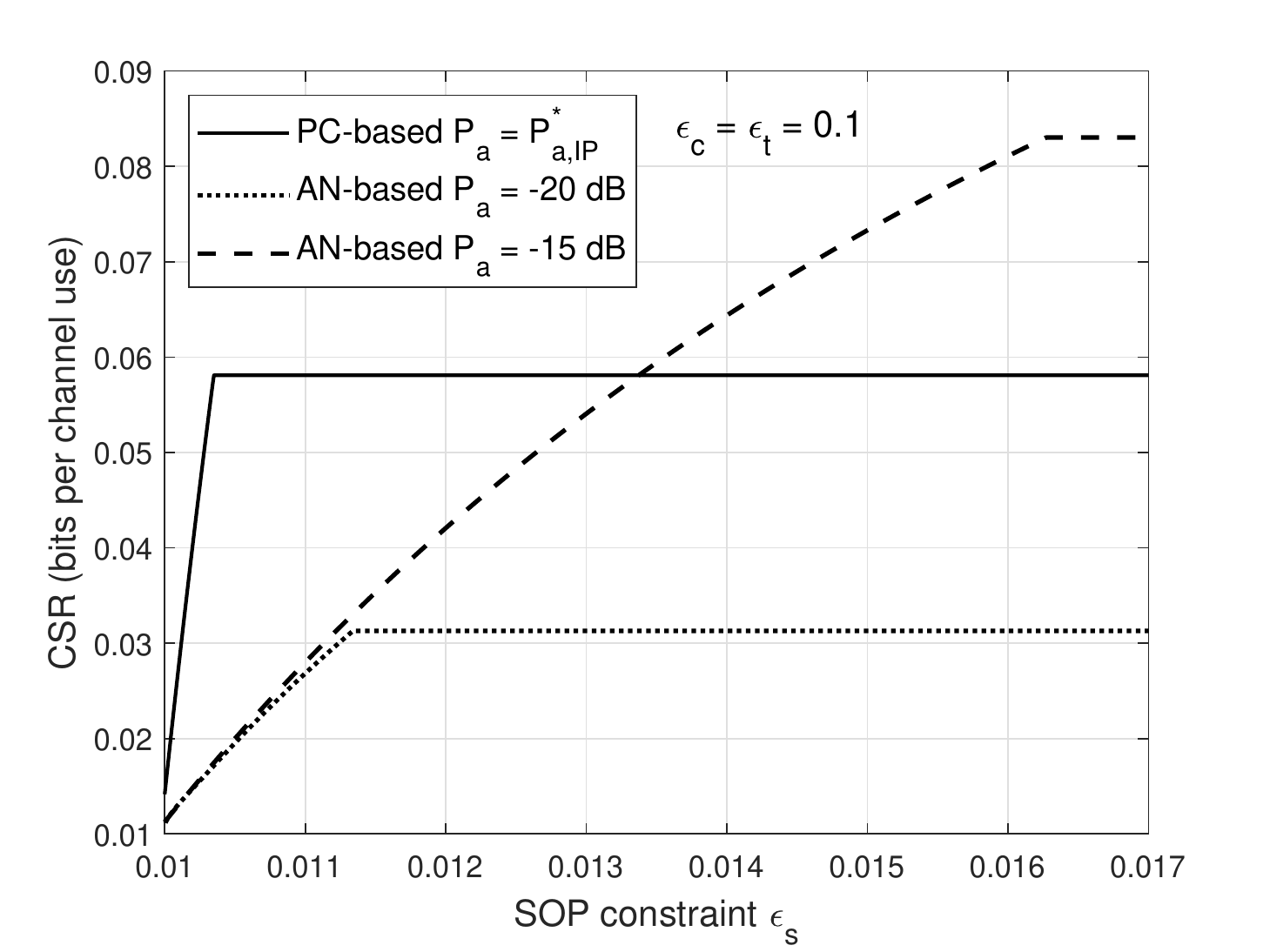}
	\label{Ind_Rcs_Es}}
\subfigure[Friend relationship scenario.]{
	\includegraphics[width=0.4\textwidth]{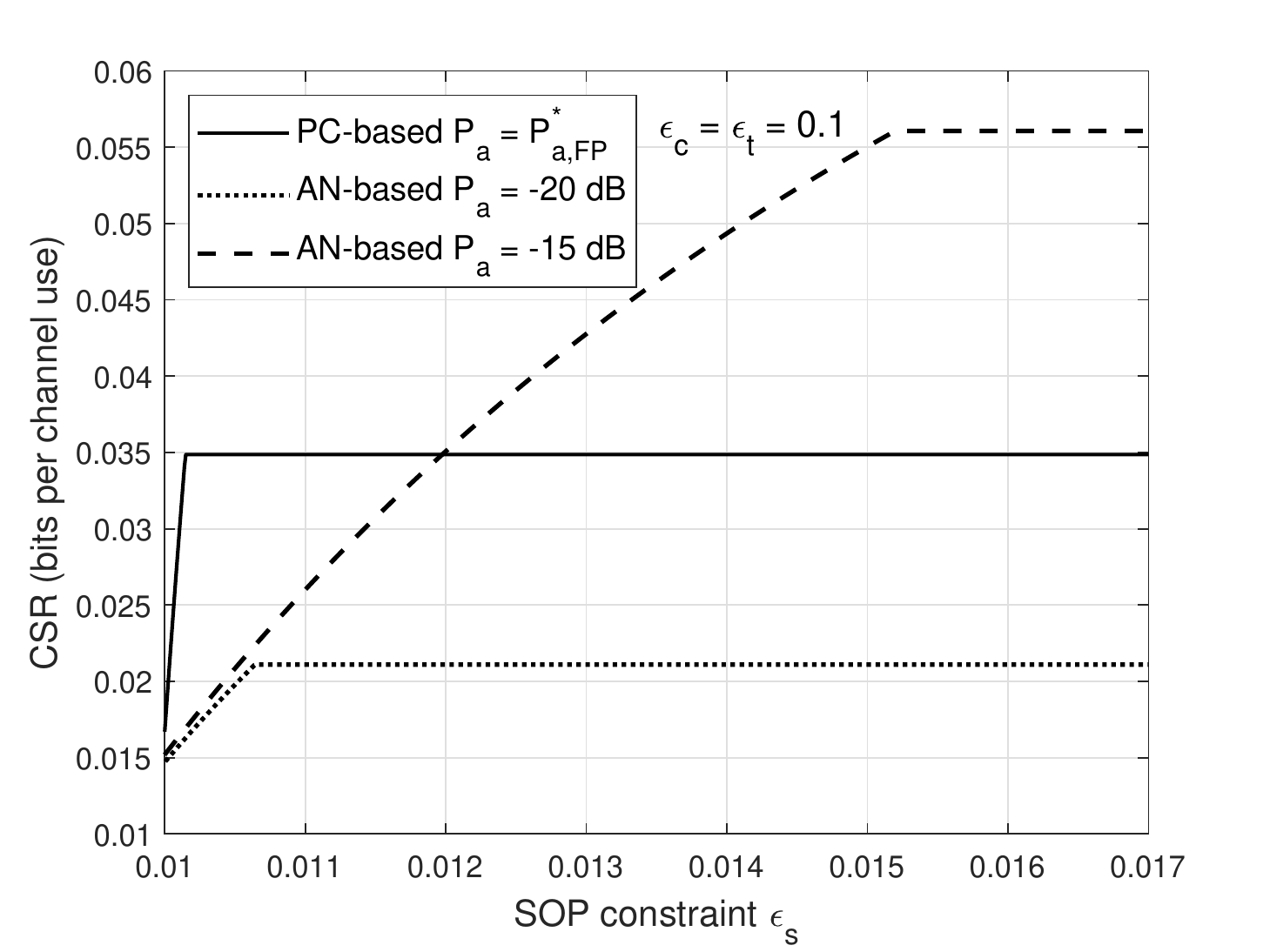}
	\label{Fri_Rcs_Es}}
\caption{Comparisons of the CSR performances in the PC-based and AN-based transmission schemes ($R_{cs}$ vs. $\epsilon_s$).}
\label{transmission_Es}
\end{figure}

\begin{figure}[!t]
\centering
\subfigure[Independence relationship scenario.]{
	\includegraphics[width=0.4\textwidth]{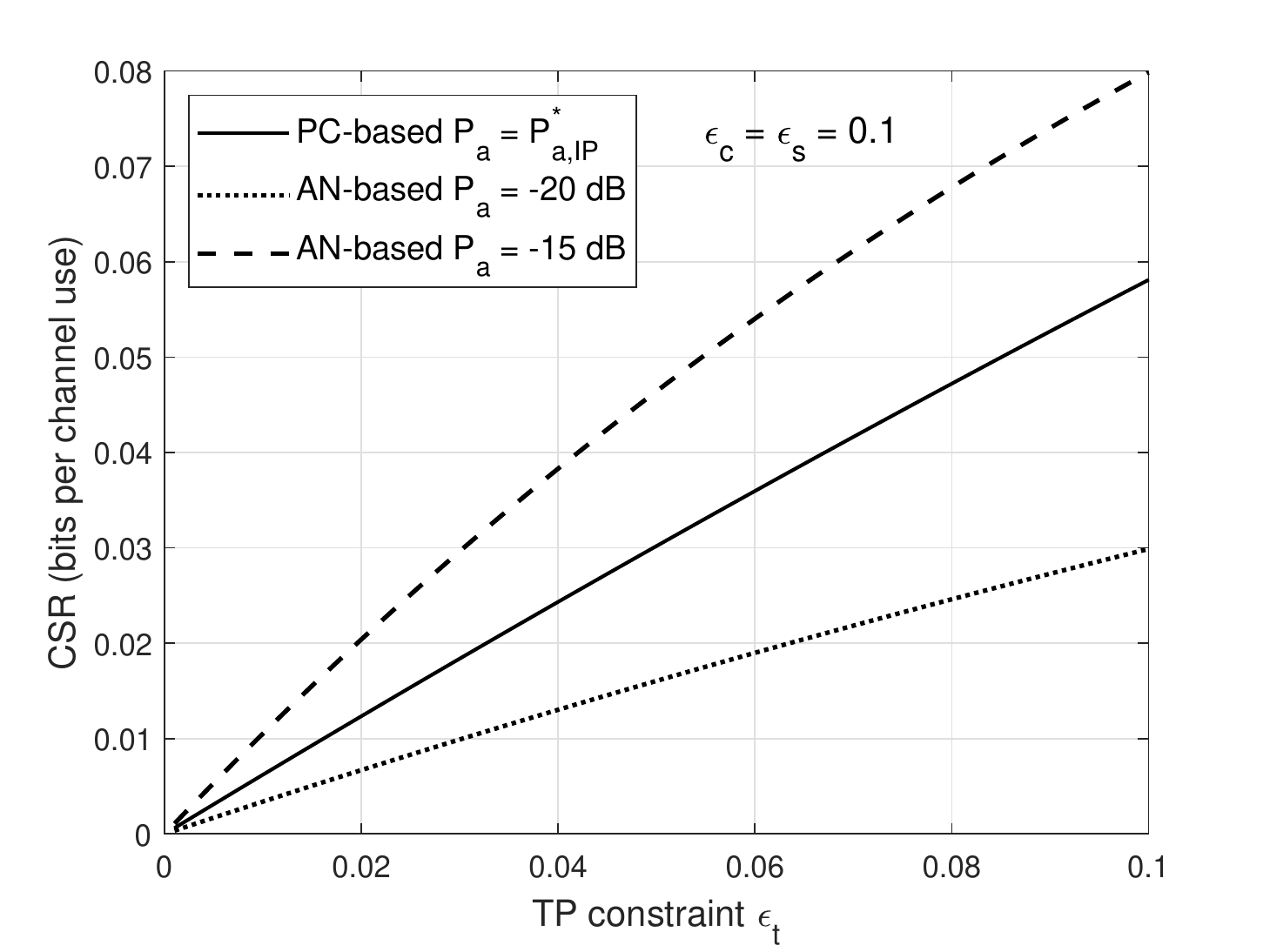}
	\label{Ind_Rcs_Et}}
\subfigure[Friend relationship scenario.]{
	\includegraphics[width=0.4\textwidth]{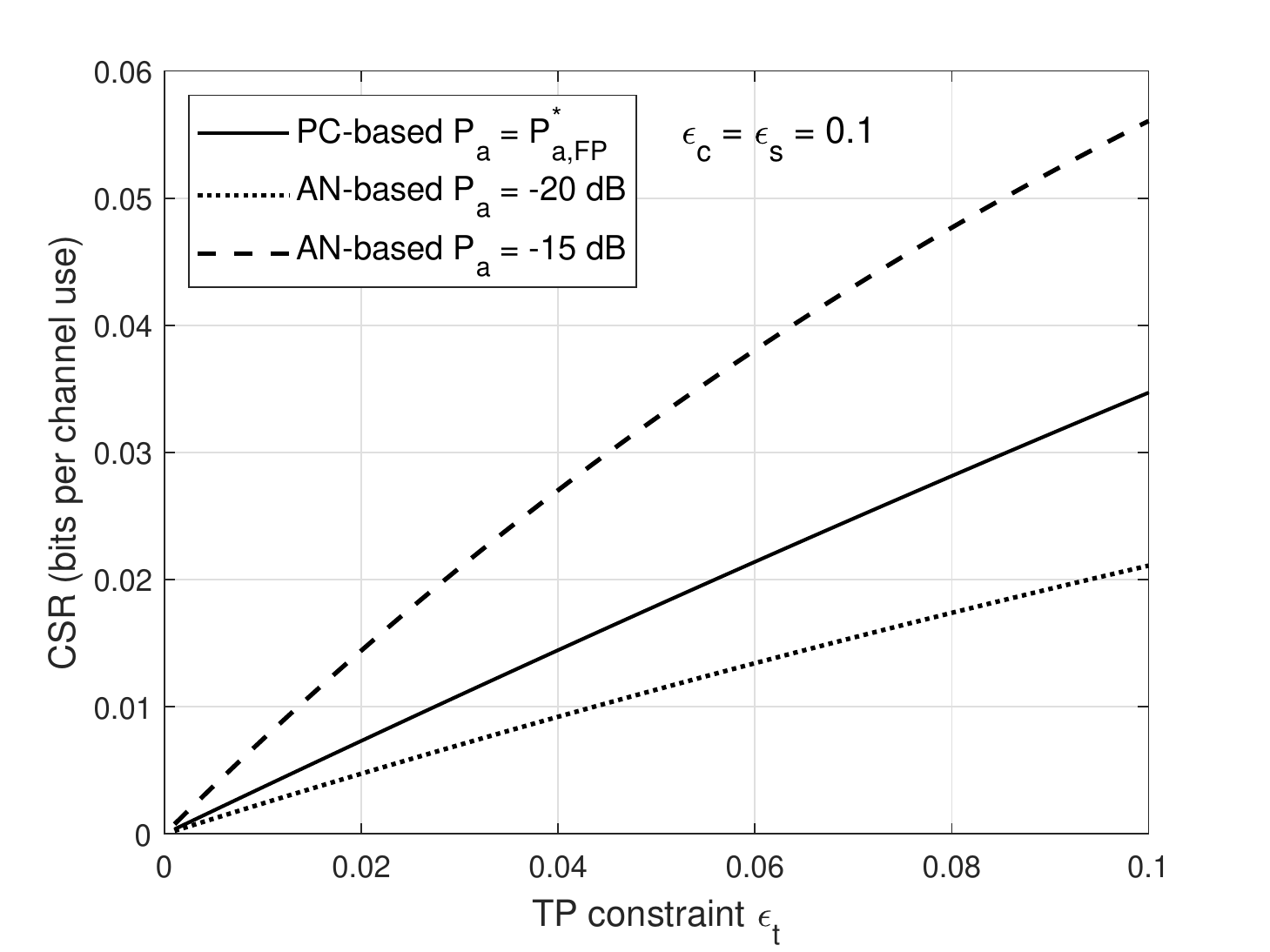}
	\label{Fri_Rcs_Et}}
\caption{Comparisons of the CSR performances in the PC-based and AN-based transmission schemes ($R_{cs}$ vs. $\epsilon_t$).}
\label{transmission_Et}
\end{figure}

Finally, we compare the PC-based transmission scheme and the AN-based transmission scheme in terms of the CSR performance. To do so, we show $R_{cs}$ vs. $\epsilon_c$ in Fig. \ref{transmission_Ec} (resp. $R_{cs}$ vs. $\epsilon_s$ in Fig. \ref{transmission_Es} and $R_{cs}$ vs. $\epsilon_t$ in Fig. \ref{transmission_Et}) under both transmission schemes in the independence and friend relationship scenarios, respectively.
We set  $\epsilon_s=\epsilon_t=0.1$ in Fig. \ref{transmission_Ec}, $\epsilon_c=\epsilon_t=0.1$ in Fig. \ref{transmission_Es} and  $\epsilon_c=\epsilon_s=0.1$ in Fig. \ref{transmission_Et}.
For each figure, we consider two different settings of the transmit power of Alice $P_a$ for the AN-based scheme.
We can observe from Fig. \ref{transmission_Ec} that, in both relationship scenarios, the PC-based scheme achieves better CSR performance than the AN-based scheme, when a small transmit power (e.g., $P_a=-20$ dB) is adopted in the AN-based scheme.  
However, when the transmit power of AN-based scheme is relatively larger (e.g., $P_a=-15$ dB), the PC-based scheme achieves better CSR performance than the AN-based scheme under stringent COP constraints (e.g., less than about $0.055$ in Fig. \ref{Ind_Rcs_Ec}), while the AN-based scheme achieves better CSR performance than the PC-based scheme under less strict COP constraints. 

Similar results can be obtained from Fig. \ref{transmission_Es}, which shows that the PC-based scheme outperforms the AN-based scheme if either the transmit power of the AN-based scheme or the SOP constraint is small. Otherwise, the AN-based scheme outperforms the PC-based scheme. However, the results obtained from Fig. \ref{transmission_Et} are different. We can see from Fig. \ref{transmission_Et} that the AN-based scheme outperforms the PC-based scheme when adopting a large transmit power (i.e., $P_a=-15$ dB), while it achieves worse CSR performance than the PC-based scheme when adopting a small transmit power  (i.e., $P_a=-20$ dB).
Based on the above observations from Fig. \ref{transmission_Ec}, Fig. \ref{transmission_Es} and \ref{transmission_Et}, we can conclude that when the transmit power is not a big concern, transmitters may prefer the AN-based transmission scheme to achieve better CSR performance, especially for less strict covertness, secrecy and transmission performance constraints. On the other hand, when the transmit power is constrained (e.g., in IoT and sensor networks), the PC-based scheme is more preferable for transmitters.

\section{Conclusion} \label{6}
This paper explores a new secure wireless communication paradigm, where the physical layer security technology is applied to ensure both the covertness and secrecy of the communication. We define a novel metric of covert secrecy rate (CSR) to depict the security performance of the new paradigm, and also provide solid theoretical analysis on CSR under two transmission schemes (i.e., artificial noise (AN)-based one and power control (PC)-based one) and two detector-eavesdropper relationships (i.e., independence and friend). The results in this paper indicate that in general the CSR performance can be improved when the constraints on covertness, secrecy and transmission performance become less strict. In particular, the PC-based transmission scheme outperforms the AN-based transmission scheme in terms of the CSR performance when strict constraints are applied to the covertness, secrecy and transmission performance. On the other hand, when these constraints become less strict, the AN-based scheme may achieve better CSR performance than the PC-based one by properly adjusting the message transmit power. We expect that this work can shed light on the future studies of new secure wireless communication paradigms.


\bibliography{acl}
\bibliographystyle{IEEEtran}


%

\end{document}